%% file: reconfig.tex
\documentclass[11pt,letterpaper]{article}
\usepackage[utf8]{inputenc}

\usepackage{fullpage}
\usepackage[colorlinks=true,linkcolor=black,citecolor=black,urlcolor=black]{hyperref}

\usepackage{subfigure}

\usepackage[pdftex]{graphicx}
\usepackage{amsmath,amsfonts,amssymb,amsthm}
\usepackage{latexsym}
\usepackage{enumitem}

\usepackage{tikz-cd}
\usepackage{xcolor}

\usepackage{authblk}

\newtheorem{theorem}{Theorem}
\newtheorem{lemma}[theorem]{Lemma}
\newtheorem{proposition}[theorem]{Proposition}
\newtheorem{corollary}[theorem]{Corollary}
\newtheorem{claim}[theorem]{Claim}
\newtheorem{observation}[theorem]{Observation}
\newtheorem{conjecture}[theorem]{Conjecture}

\def\defn#1{\textit{\textbf{\boldmath #1}}}

\makeatletter
\long\def\@caption#1[#2]#3{\par\addcontentsline{\csname
  ext@#1\endcsname}{#1}{\protect\numberline{\csname 
  the#1\endcsname}{\ignorespaces #2}}\begingroup
    \@parboxrestore
    \small
    \@makecaption{\csname fnum@#1\endcsname}{\ignorespaces #3}\par
  \endgroup}
\makeatother

\newcommand{\remove}[1]{}
\newcommand{\changed}[1]{{\color{blue}#1}}
\newcommand{\mchanged}[1]{{\color{olive}#1}}

\newcommand{\anna}[1]{{\color{brown}Anna says: #1}}
\newcommand{\jaysonl}[1]{{\color{magenta}Jayson says: #1}}
\newcommand{\oswin}[1]{{\color{orange}Oswin says: #1}}
\newcommand{\therese}[1]{{\color{teal}Therese says: #1}}

\newcommand{\mati}[1]{{\color{red}Matias says: #1}}
\newcommand{\pepa}[1]{{\color{violet}Pepa  says: #1}}
\newcommand{\yushi}[1]{{\color{cyan}Yushi  says: #1}}

\title{Reconfiguration of Non-crossing Spanning Trees}
\author[1]{Oswin Aichholzer}
\affil[1]{University of Technology Graz, Austria. {\tt oaich@ist.tugraz.at}}
\author[2]{Brad Ballinger}
\affil[2]{Cal Poly Humboldt, USA. {\tt brad@humboldt.edu}}
\author[3]{Therese Biedl}
\affil[3]{University of Waterloo, Canada. {\tt \{biedl,alubiw,jayson.lynch\}@uwaterloo.ca}}
\author[4]{Mirela Damian}
\affil[4]{Villanova University, USA. {\tt mirela.damian@villanova.edu}}
\author[5]{Erik D. Demaine}
\affil[5]{MIT Computer Science and Artificial Intelligence Laboratory, USA. \protect{\url{edemaine@mit.edu}}}
\author[6]{Matias Korman}
\affil[6]{Siemens Electronic Design Automation, USA. {\tt matias\_korman@mentor.com}}
\author[3]{Anna Lubiw}
\author[3]{Jayson Lynch}
\author[7]{Josef Tkadlec}
\affil[7]{Harvard Unversity, USA. {\tt tkadlec@math.harvard.edu}}
\author[8]{Yushi Uno}
\affil[8]{Osaka Metropolitan University, Japan. {\tt yushi.uno@omu.ac.jp}}

\remove{
\author[1]{Oswin Aichholzer}
\affil[1]{University of Technology Graz, Austria. {\tt oaich@ist.tugraz.at}}
\author[2]{Brad Ballinger}
\affil[2]{Cal Poly Humboldt, USA. {\tt brad@humboldt.edu}}
\author[3]{Therese Biedl}
\affil[3]{University of Waterloo, Canada. {\tt \{biedl,alubiw\}@uwaterloo.ca}}
\author[4]{Mirela Damian}
\affil[4]{Villanova University, USA. {\tt mirela.damian@villanova.edu}}
\author[5]{Erik D. Demaine}
\affil[5]{Erik's at MIT}
\author[6]{Matias Korman}
\affil[6]{Siemens Electronic Design Automation, USA. {\tt matias\_korman@mentor.com}}
\author[3]{Anna Lubiw}
\author[3]{Jayson Lynch}
\author[7]{Josef Tkadlec}
\affil[7]{Harvard Unversity, USA. {\tt tkadlec@math.harvard.edu}}
\author[8]{Yushi Uno}
\affil[8]{Osaka Metropolitan University, Japan. {\tt yushi.uno@omu.ac.jp}}
}

\begin{document}

\maketitle

\begin{abstract}
For a set $P$ of $n$ points in the plane in general position, 
a \defn{non-crossing spanning tree} is a spanning tree of the points 
where every edge is a straight-line segment between a pair of points and
no two edges intersect except at a common endpoint. 
We study the problem of \defn{reconfiguring} one non-crossing spanning tree of $P$ to another 
using a sequence of \defn{flips} where each flip removes one edge 
and adds one new edge so that the result is again a non-crossing spanning tree of $P$. 
There is a known upper bound of $2n-4$ flips [Avis and Fukuda, 1996] and a lower bound of $1.5n - 5$ flips. 

We give a reconfiguration algorithm that uses at most $2n-3$ flips but reduces that to $1.5n-2$ flips when one tree is a path and either: the points are in convex position; or the path is monotone in some direction. 
For points in convex position, we prove an upper bound of 
$2d - \Omega(\log d)$ where $d$ is half the size of the symmetric difference between the trees. 
We also examine whether the \defn{happy edges} (those common to the initial and final trees) need to flip, 
and we find exact minimum flip distances for small point sets using exhaustive search.
\end{abstract}

\section{Introduction}
\label{sec:introduction}
Let $P$ be a set of $n$ points in the plane in general position.
A \defn{non-crossing spanning tree} is a spanning tree of $P$
whose edges are straight line segments between pairs of points such that no two edges intersect except at a common endpoint.
A \defn{reconfiguration step} or \defn{flip} removes one edge of a non-crossing spanning tree and adds one new edge so that the result is again a non-crossing spanning tree of $P$.
We study the problem of
\defn{reconfiguring} one non-crossing spanning of $P$ to another
via a sequence of flips.


Researchers often consider three problems about reconfiguration, which are most easily expressed in terms of the \defn{reconfiguration graph} that has a vertex for each configuration (in our case, each non-crossing spanning tree) and an edge for each reconfiguration step.  The problems are:
(1) connectivity of the reconfiguration graph---is reconfiguration always possible? 
(2) diameter of the reconfiguration graph---how many flips are needed for reconfiguration in the worst case? and (3) distance in the reconfiguration graph---what is the complexity of finding the minimum number of flips to reconfigure between two given configurations? 

For reconfiguration of non-crossing spanning trees,
Avis and Fukuda~\cite[Section 3.7]{avis1996reverse} proved 
that reconfiguration is always possible, 
and that at most $2n-4$ flips are needed. 
Hernando et al.~\cite{hernando1999geometric} proved a lower bound of $1.5n - 5$ flips.
Even for the special case of points in convex position, there are no better upper or lower bounds known. 


Our two main results 
make some progress in reducing the diameter upper bounds.

{\bf (1)} For points in general position, we give a reconfiguration algorithm 
that uses at most $2n-3$ flips but reduces that to $1.5n-2$ flips
in two cases: (1) when the points are in convex position and one tree is a path; (2) for general point sets when one tree is a monotone path. 
The algorithm first flips one tree
to a downward tree
(with each vertex connected to a unique higher vertex)
and the other
tree to an upward tree (defined symmetrically)
using $n-2$ flips---this is where we save when one tree is a path.  After that, we give an algorithm to  flip from a downward tree $T_D$ to an upward tree $T_U$ using at most $n-1$ 
``perfect'' flips each of which removes an edge of $T_D$ and adds an edge of $T_U$. The algorithm is simple to describe, but proving that intermediate trees remain non-crossing is non-trivial.
We also show that $1.5n - 5$ flips
may be required, even in the two special cases.
See Section~\ref{sec:2phase}.

{\bf (2)} For points in convex position, we improve the upper bound on the number of required flips to
$2d - \Omega(\log d)$ where $d$ is half the size of the symmetric difference between the trees.
So $d$ flips are needed in any flip sequence, and $2d$ is an upper bound.
The idea is to find an edge $e$ of one tree that is crossed by at most (roughly) $d/2$ edges of the other tree, 
flip all but one of the crossing edges temporarily to the convex hull (this will end up costing 2 flips per edge), and then flip the last crossing edge to $e$. Repeating this saves us one flip, compared to the $2d$ bound, for each of the (roughly) $\log d$ repetitions. 
See Section~\ref{sec:improved-bound}.

Notably, neither of our algorithms uses the common---but perhaps limited---technique of 
observing that the diameter is at most twice the radius, and bounding the radius of the reconfiguration 
graph by identifying a ``canonical'' configuration that every other configuration can be flipped to.  Rather, our
algorithms find reconfiguration sequences tailored to the specific input trees.

In hopes of
making further progress on the diameter and distance problems,
we address the question of \emph{which} edges need to be involved in a minimum flip sequence from an initial non-crossing spanning tree $T_I$ to a final non-crossing spanning tree $T_F$. 
We say that the edges of $T_I \cap T_F$ are \defn{happy} edges, and we formulate the \defn{Happy Edge Conjecture} that for points in convex position, there is a minimum flip sequence that never flips happy edges. 
We prove the conjecture for happy convex hull edges. 
See Section~\ref{sec:happy-edges}.
More generally, we say that a reconfiguration problem has the ``happy element property'' if elements that are common to the initial and final configurations can remain fixed in a minimum flip sequence.  
%
Reconfiguration problems that satisfy the happy element property seem easier.
For example, the happy element property holds for reconfiguring 
spanning trees in a graph (and indeed for matroids more generally), and the distance problem is easy.
On the other hand, 
the happy element property
fails for reconfiguring triangulations of a point set in the plane, and for the problem of \emph{token swapping} on a tree~\cite{biniaz2019token}, 
and in both cases, 
this is the key to constructing gadgets 
to prove that the distance problem is NP-hard~\cite{LUBIW201517,pilz2014flip,aichholzer2021hardness-tokens}.
As an aside, we note that 
for reconfiguring triangulations of a set of points in convex position---where the distance problem is the famous open question of rotation distance in binary trees---the happy element property holds~\cite{sleator1988}, which may be why no one has managed to prove that the distance problem is NP-hard. 
%

Finally, 
we implemented a combinatorial search program to compute the diameter (maximum reconfiguration distance between two trees) and radius of the reconfiguration graph for points in convex position.
For $6 \leq n \leq 12$ the diameter is $\lfloor 1.5n-4 \rfloor$ and the radius is $n-2$. In addition we provide the same information for the special case when the initial and final trees are non-crossing spanning paths, though intermediate configurations may be more general trees. We also verify 
the Happy Edge Conjecture for $n \leq 10$ points in convex position.
See Section~\ref{sec:exhaustive}.


\subsection{Background and Related Results}
\label{sec:background}


Reconfiguration is about changing one structure to another, either through continuous motion or through discrete changes. 
In mathematics, the topic has a vast and deep history, for example in knot theory, 
and the study of bounds on the simplex method for linear programming.  
Recently, reconfiguration via discrete steps has become a focused research area, see the surveys by van den Heuvel~\cite{vandenheuvel2013complexity} and Nishimura~\cite{nishimura2018introduction}.  
Examples include sorting a list by swapping pairs of adjacent elements, solving a Rubik’s cube, or changing one colouring of a graph to another. 
With discrete reconfiguration steps, the reconfiguration graph is well-defined.  Besides questions of diameter and distance in the reconfiguration graph, there is also research on enumeration via a Hamiltonian cycle in the reconfiguration graph, see the recent survey~\cite{mutze2022combinatorial}, and on mixing properties to find random configurations, see~\cite{randall2006rapidly}.

Our work is about reconfiguring one graph to another.
Various reconfiguration steps have been considered, for example exchanging one vertex for another (for reconfiguration of paths~\cite{bonsma2013complexity}, independent sets~\cite{ito2020parameterized,avis2022reconfiguration}, etc.), or exchanging multiple edges
(for reconfiguration of matchings~\cite{bonamy_et_al:LIPIcs:2019:11024}).  
However, we concentrate on elementary steps (often called \emph{edge flips}) that exchange  one edge for one other edge.
A main example of this is reconfiguring one spanning tree of a graph to another, which can always be accomplished by ``perfect'' flips that add an edge of the final tree and delete an edge of the initial tree---more generally, such a perfect exchange sequence is possible when reconfiguring bases of a matroid.
%

Our focus is on geometric graphs whose vertices are points in the plane
and whose edges are non-crossing line segments between the points.
In this setting, one well-studied problem is reconfiguring between triangulations of a point set in the plane, see the survey by Bose and Hurtado~\cite{BOSEHURTADO200960}. Here, a flip replaces 
an edge in a convex quadrilateral by the other diagonal of the quadrilateral. 
For the special case of $n$ points in convex position this is equivalent to rotation of an edge in a given (abstract) rooted binary tree, which is of interest in the study of splay trees, and the study of phylogenetic trees in computational biology.
While an upper bound for the reconfiguration distance of $2n-10$ is known to be tight 
for $n > 12$~\cite{pournin2014diameter,sleator1988},
the complexity of computing the shortest distance between 
two triangulations of a convex point set
(equivalently between two given binary trees) is still unknown. See~\cite{amp-fdtsp-15,LUBIW201517,pilz2014flip} for related hardness results for the flip-distance of triangulations of point sets and simple polygons.


Another well-studied problem for geometric graphs is  reconfiguration of non-crossing perfect matchings.  Here a flip typically exchanges matching and non-matching edges in  non-crossing cycles, and there are results for a single cycle of unbounded length~\cite{houle2005graphs}, and for multiple cycles~\cite{AICHHOLZER2009617,razen2008lower}.
For points in convex position, the flip operation on a single cycle of length 4 suffices to connect the reconfiguration graph~\cite{Hernando02graphsof}, but this is open for general point sets.
For a more general flip operation that connects any two disjoint matchings whose union is non-crossing, the reconfiguration graph is connected for points in convex position~\cite{aam-dcgnc-15} (and, notably, the diameter is less than twice the radius), but connectivity is open for general point sets~\cite{AICHHOLZER2009617,IST-DCGM-2013}. 


The specific geometric graphs we study are non-crossing (or ``plane'') spanning trees of a set of points in the plane. 
For points in convex position, these have been explored for enumeration~\cite{noy1998enumeration}, and for a duality with quadrangulations and consequent lattice properties, e.g., see~\cite{apostolakis2021non} and related literature.  

For non-crossing spanning trees of a general point set in the plane,
%
there 
are
several basic reconfiguration operations that can be used to transform these trees into each other. The one we use in this work is the simple {\it edge exchange} of an edge $e$ by an edge $e'$ as described above. If 
we require that
$e$ and $e'$ do not cross, then this operation is called a {\it compatible edge exchange}. Even more restricted is an {\it edge rotation}, where 
$e=uv$ and $e'=uw$ 
share a common vertex $u$.
If the triangle $u,v,w$ is not intersected by an edge of the two involved trees, then this is called an {\it empty triangle edge rotation}. The most restricted operation is the {\it edge slide} (see also Section~\ref{sec:edge-slides}) where the edge $vw$ has to exist in both trees. The name comes from viewing this transformation as sliding one end of the edge $e$ along $vw$  to $e'$ (and rotating the other end around $u$) and at no time crossing through any other edge of the trees.
For an overview and detailed comparison of the described operations see Nichols et al.~\cite{nichols2020transition}.

It has been shown that the reconfiguration graph of non-crossing spanning trees is even connected
for the ``as-local-as-possible'' edge slide operation~\cite{aah-sstft-00}. See also~\cite{aichholzerreinhardt2007155} where a tight bound of $\Theta(n^2)$ steps for the diameter is show. This implies that also for the other reconfiguration operations the flip graph is connected.
For edge exchange, compatible edge exchange, and edge rotation a linear upper bound for the diameter is known~\cite{avis1996reverse}, while for empty triangle edge rotations an upper bound of $O(n \log n)$ has been shown recently~\cite{nichols2020transition}.
For all operations (except edge slides) the best known lower bound is $1.5n-5$~\cite{hernando1999geometric}.

There 
are several variants for the reconfiguration of spanning trees. For example, the operations can be performed in parallel, that is, as long as the exchanges (slides etc) do not interfere with each other, they can be done in one step; see~\cite{nichols2020transition} for an overview of results.
In a similar direction of a more global operation we can say that two non-crossing spanning trees are compatible, if their union is still crossing free. A single reconfiguration step then transforms one such tree into the other. A lower bound of $\Omega(\log n / \log \log n)$~\cite{buchin2009transforming} and an upper bound of $O(\log n)$~\cite{aah-sstft-00} for the diameter of this reconfiguration graph has been shown.

Another variation is that the edges are labeled (independent of the vertex labels), and if an edge is exchanged the replacement edge takes that label. In this way two geometrically identical trees can have a rather large transformation distance.
For labeled triangulations 
there is a good characterization of when reconfiguration is possible, and a polynomial bound on the number of steps required~\cite{lubiw2019proof}.

The reconfiguration of non-crossing spanning paths (where each intermediate configuration must be a path) has also been considered. For points in convex position, the diameter of the reconfiguration graph is $2n-6$ for $n \ge 5$~\cite{akl2007planar,chang2009diameter}. Surprisingly, up to now it remains an open problem if the reconfiguration graph of non-crossing spanning paths is connected for general point sets~\cite{aklmmopv-fpsp-22}.





\subsection{Definitions and Terminology}

Let $P$ be a set of $n$ points in 
\emph{general position}, meaning that no three points are collinear. 
The points of $P$ are in \defn{convex position} if the boundary of the convex hull of $P$ contains all the points of $P$. 
%
An \defn{edge} is a line segment joining two points of $P$, and a \defn{spanning tree} $T$ of $P$ is a set of $n-1$ edges that form a tree.  Two edges \defn{cross} if they intersect but do not share a common endpoint. 
A \defn{non-crossing spanning tree} is a spanning tree such that no two of its edges cross. 
When we write ``two non-crossing spanning trees,'' we mean that each tree is non-crossing but we allow edges of one tree to cross edges of the other tree.

We sometimes consider the special case where a non-crossing spanning tree of $P$ is a path.  A path is \defn{monotone} if there is a direction in the plane such that 
the order of points along the path matches the order of points in that direction.

For a spanning tree of a graph, 
a \defn{flip} removes one edge and adds one new edge to obtain a new spanning tree, i.e., spanning trees $T$ and $T'$ are related by a flip if $T' = T \setminus \{ e \} \cup \{e'\}$, where $e \in T$ and $e' \notin T$.  
The same definition applies 
to non-crossing spanning trees: If $T$ is a non-crossing spanning tree of $P$, and $T' = T \setminus \{ e \} \cup \{e'\}$ is a non-crossing spanning tree of $P$, then we say that $T$ and $T'$ are related by a \defn{flip}.  We allow $e$ and $e'$ to cross.

Let $T_I$ and $T_F$ be initial and final non-crossing spanning trees of $P$.  A \defn{flip sequence} from $T_I$ to $T_F$ is a sequence of flips that starts with $T_I$ and ends with $T_F$ and such that each intermediate tree is a non-crossing spanning tree.
We say that $T_I$ can be reconfigured to $T_F$ 
\defn{using $k$ flips} (or ``in $k$ steps'') if there is a reconfiguration sequence of length at most $k$.
The \defn{flip distance} from $T_I$ to $T_F$ is the minimum length of a flip sequence.

The edges of $T_I \cap T_F$ are called \defn{happy} edges.  
Thus, $T_I \cup T_F$ consists of the happy edges together with the symmetric difference $(T_I \setminus T_F) \cup (T_F \setminus T_I)$.  We have $|T_I \setminus T_F| = |T_F \setminus T_I|$.  A flip sequence of length $|T_I \setminus T_F|$ is called a \defn{perfect flip sequence}.  In a perfect flip sequence, every flip removes an edge of $T_I$ and adds an edge of $T_F$---these are called \defn{perfect flips}.
\section{A two-phase reconfiguration approach}
\label{sec:2phase}

In this section we give a new algorithm to reconfigure between two non-crossing spanning trees on $n$ points using at most $2n-3$ flips.  This is basically the same as the upper bound of $2n-4$ originally achieved by Avis and Fukuda~\cite{avis1996reverse}, but the advantage of our new algorithm is that it gives a bound of $1.5n-2$ flips when one tree is a path and either: (1) the path is monotone; or (2) the points are in convex position.
Furthermore, for these two cases, we show a lower bound of $1.5n -5$ flips, so the bounds are tight up to the additive constant.

Before proceeding, we mention one way in which our upper bound result differs from some other reconfiguration bounds.
Many of those bounds (i.e., upper bounds on the diameter $d$ of the reconfiguration graph) are actually bounds on the \emph{radius} $r$ of the reconfiguration graph. 
The idea is to identify a 
``canonical'' configuration and  prove that its distance to any other configuration is at most $r$, thus proving that the diameter $d$ is at most $2r$.  For example, Avis and Fukuda's $2n$  bound is achieved via a canonical star centered at a convex hull point. As another example, 
the bound of  $O(n^2)$ flips between triangulations of a point set can be proved using the Delaunay triangulation as a canonical configuration, and 
the bound of $2n$ flips for the special case of points in convex position uses a canonical star triangulation~\cite{sleator1988}. 
For some reconfiguration graphs  $d$ is equal to $2r$ (e.g., for the Rubik's cube, because of the underlying permutation group). However, in general, $d$ can be less than $2r$, in which case, 
using a canonical configuration will not give the best diameter bound.
Indeed, our result does not use a canonical configuration, and we do not bound the radius of the reconfiguration graph. 


\remove{Introduce this section.  Most bounds on the number of  reconfiguration steps are obtained by reconfiguring the initial and final configurations to some ``canonical'' configuration.  
For example, the $2n$ upper bound of Avis and Fukuda used as a canonical configuration a  
star centered at a convex hull point (Avis and Fukuda).  Also, for convex point sets, one can obtain the $2n$ upper bound using a path on the convex hull as the canonical configuration.  Using a single canonical configuration is fine if the diameter and radius (is it called eccentricity?) of the flip graph are the same asymptotically---but why should they be? and then we will need different techniques.  Note that our proof of a better bound in the  previous section did not use a canonical configuration.

Here we use a different technique . . . 

\therese{This also needs an outlook---what is the main result of this section?   
}
}

Our algorithm has two phases.  In the first phase, we reconfigure the input trees in a total of at most $n-2$ flips so that one is ``upward'' and one is ``downward'' (this is where we save 
if one tree is a path).  In the second phase we show that an upward tree can be reconfigured to a downward tree in $n-1$ flips.  We begin by defining these terms.

Let $P$ be a set of $n$ points in general position.  Order the points $v_1, \ldots, v_n$ by increasing $y$-coordinate
(if necessary, we  slightly perturb the point set to ensure that no two $y$-coordinates are identical). 
Let $T$ be a non-crossing spanning tree of $P$.
\remove{For a node $v_i$, the \defn{upward degree} $d^+(v_i)$ is the number of edges in $T$ connecting $v_i$ to a node with a higher $y$-coordinate (i.e, the number of edges pointing \defn{upward}). Similarly, the \defn{downward degree} $d^-(v_i)$ is the number of edges in $T$ connecting $v_i$ to a node with a lower $y$-coordinate (pointing \defn{downward}).}
Imagining the edges as directed upward, 
we call a vertex $v_i$ a \defn{sink} if there are no edges in $T$ connecting $v_i$ to a higher vertex $v_j, j > i$, and we call $v_i$ a \defn{source} if there are no edges connecting $v_i$ to a lower vertex $v_k, k < i$. 
We call $T$ a \defn{downward tree} if it has only one sink (which must then be $v_n$) and we call $T$ an \defn{upward tree} if it has only one source (which must then be $v_1$).
Observe that in a downward tree every vertex except $v_n$ has exactly one edge connected to a higher vertex, and in an upward tree every vertex except $v_1$ has exactly one edge connected to a lower vertex.
\subsection{Phase 1: Reconfiguring to upward/downward trees}
\label{sec:phase1}

We first bound the number of flips needed to reconfigure a single tree $T$ to be upward or downward.
If a tree has $t$ sinks, then we need at least $t-1$ flips to reconfigure it to a downward tree---we show that $t-1$ flips suffice.  Note that $t$ is at most $n-1$ since $v_1$ cannot be a sink (this bound is realized by a  star at $v_1$).

\begin{theorem}\label{theo:canonicalize}
\remove{
Let $T$ be a non-crossing spanning tree.
$T$ can be reconfigured to a downward tree with $e^+(T) \le n-2$ flips. 
$T$ can be reconfigured to an upward tree with $e^-(T) \le n-2$ flips. 
}
Let $T$ be a non-crossing spanning tree with $s$ sources and $t$ sinks.
$T$ can be reconfigured to a downward tree with $t-1 \le n-2$ flips. 
$T$ can be reconfigured to an upward tree with $s-1 \le n-2$ flips. 
Furthermore, these reconfiguration sequences do not flip any edge of the form
$v_iv_{i+1}$ where $1 \le i < n$. 
%
%
%
\label{thm:tree-downward}
\end{theorem}

\begin{proof}
We give the proof for a downward tree, since the other case is symmetric. 
The proof is by induction on $t$.  In the base case, $t=1$ and the tree is downward and so no flips are needed.
Otherwise, let $v_i, 1 < i < n$ be a sink. 
The plan is to decrease $t$ by adding an edge going upward from $v_i$ and removing {some edge $v_k v_l, k<l$} 
from the resulting cycle  while ensuring that $v_k$ does not become a sink. 

If there is an edge $v_i v_j$ that does not cross any edge of $T$, we say that $v_i$ \emph{sees} $v_j$.
We argue that $v_i$ sees some vertex $v_j$ with $j > i$.
If $v_i$ sees $v_n$, then choose $j=n$.  Otherwise the upward ray directed from $v_i$ to $v_n$ hits some edge $e$ before it reaches $v_n$. Continuously rotate the ray around $v_i$ towards the higher endpoint of $e$ until the ray reaches the endpoint or is blocked by some other vertex.
%
%
In either case this gives us a vertex $v_j$ visible from $v_i$ and higher than $v_i$.
For example, in Figure~\ref{fig:phase1}, the sink $v_5$ sees $v_7$.

\begin{figure}[htp]
    \centering
    \includegraphics[page=1,width=0.7\textwidth]{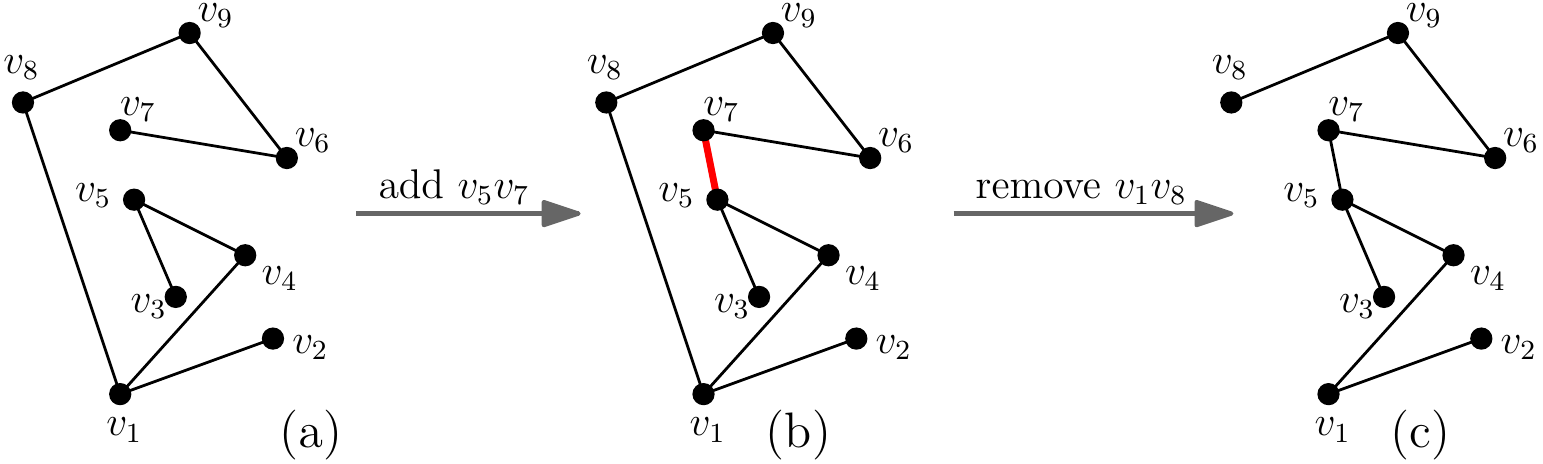}
    \caption{A flip that removes $v_5$ from the set of sinks.
    }
    \label{fig:phase1}
\end{figure}

Adding the edge $v_i v_j$ to $T$ creates a cycle.  Let $v_k$ be the lowest vertex in the cycle.  Then $v_k$ has two upward edges, and in particular, $k < i$.  Remove the edge $v_k v_l$ that goes higher up.  Then $v_k$ does not become a sink, and furthermore, if the edge $v_k v_{k+1}$ is in $T$, then we do not remove it. 
\end{proof}

Since no vertex is both a sink and a source, any tree has $s+t \le n$, which yields the following result that will be useful later on when we reconfigure between a path and a tree.
\begin{corollary}
\label{cor:up-or-down}
Let $T$ be a non-crossing spanning tree.  Then either $T$ can be reconfigured to a downward tree in $0.5n - 1$ flips or $T$ can be reconfigured to an upward tree in $0.5n - 1$ flips.  Furthermore, these reconfiguration sequences do not flip any edge of the form 
$v_iv_{i+1}$ where $1 \le i < n$.
\end{corollary}




We next bound the number of flips needed to reconfigure two given trees into 
\defn{opposite trees}, 
meaning that one tree is upward and one is downward.
By Theorem~\ref{thm:tree-downward}, we can easily do this in at most $2n - 4$ flips (using $n-2$ flips to reconfigure each tree independently). We now show that $n-2$ flips suffice to reconfigure the two trees into opposite trees. 


\begin{theorem} 
Given two non-crossing spanning trees on the same point set,
we can flip them into opposite trees
in at most $n-2$ flips.
\label{thm:trees-opposite}
\remove{
\pepa{3 minor comments on the wording: (i) maybe clarify that each tree is non-crossing individually? E.g. ``two spanning trees, each non-crossing''. 
\anna{That's awkward. But I  think it's clear enough to say two non-crossing trees.}
(ii) the ``on the same point set'' could be omitted. 
\anna{Hmm.  Here I think this matters, to be precise in the theorem statement.}
(iii) strictly speaking, ``flip them in opposite directions'' is undefined (and maybe could be confusing?). I'd be fine with ``Turn them into opposite trees'' (ideally after we define opposite).
\anna{fixed, thanks.} Another (much) more verbose alternative is ``turn one of them into a downward tree and the other one into an upward tree in at most $n-2$ total flips.''}
\mati{We should define the term "plane spanning tree" early on and use it along the whole paper (or simply say that all trees we look at are plane and spanning). This avoids the issue of whether the noncrossing applies to each tree individually or as a group}
\anna{I favour ``non-crossing'' -- but I put the comment in the Definitions part so others can comment.  BTW, the issue of whether the condition applies to each tree separately or both together is the same whether we use non-crossing or plane.}
}
\end{theorem}
\begin{proof}
\remove{
Consider flipping $T_1$ upward or flipping $T_1$ downward.  By Lemma~\ref{lem:tree-up-down}, the total number of flips for both alternatives is at most $n-2$.  Similarly for $T_2$. Thus the four possibilities ($T_1$ upward, $T_1$ downward, $T_2$ upward, $T_2$ downward) take at most $2n-4$ flips. This implies that one of the pairs ($T_1$ upward, $T_2$ downward) or ($T_1$ downward, $T_2$ upward) must take at most $n-2$ flips.
}
Let the trees be $T_1$ and $T_2$, and
let $s_i$ and $t_i$ be the number of sources and sinks of $T_i$, for $i=1,2$.  
Since $s_i + t_i \le n$, 
we have
$s_1+t_1+s_2+t_2\leq 2n$.
This implies that 
$s_1+t_2\leq n$ or $t_1+s_2\leq n$.
In the former case use Theorem~\ref{thm:tree-downward} to flip $T_1$ upward and $T_2$ downward in $s_1-1+t_2-1\leq n-2$ flips; otherwise flip $T_1$ downward and $T_2$ upward in $t_1-1+s_2-1\leq n-2$ flips.
\end{proof}

\input{proof.tex}

\subsection{Two-phase reconfiguration algorithm}
We can now combine the results of Sections~\ref{sec:phase1} and~\ref{sec:phase2} to develop a new two-phase reconfiguration algorithm between two non-crossing spanning trees $T_I$ and $T_F$: 
\begin{enumerate}
\item In the first phase we reconfigure $T_I$ into $T'_I$ and $T_F$ into $T'_F$ such that $T'_I$ and $T'_F$ are opposite trees (one upward and one downward), using
Theorem~\ref{thm:trees-opposite}.
\item In the second phase we reconfigure $T'_I$ into $T'_F$ using only perfect flips, 
as given by Theorem~\ref{thm:trees-up-down}. 
(Note, however, that the happy edges in $T'_I$ and $T'_F$ may differ from the ones in $T_I$ and $T_F$, since the first phase does not preserve happy edges). 
\end{enumerate}

Finally, we concatenate the  reconfiguration sequences from $T_I$ to $T'_I$, from $T'_I$ to $T'_F$, and the reverse of the sequence from $T_F$ to $T'_F$.

\begin{theorem}
If $T_I$ and $T_F$ are non-crossing spanning trees on a general point set, then the algorithm presented above reconfigures $T_I$ to $T_F$ in at most 
$2n-3$ flips. 
\end{theorem}
\begin{proof}
By Theorem~\ref{thm:trees-opposite}, the first phase of the algorithm takes at most $n-2$ flips.
By Theorem~\ref{thm:trees-up-down}, the second phase uses 
at most $n-1$ flips.  
It follows that the total number of flips is at most $2n-3$. 
\end{proof}

\begin{theorem}
\label{tree-and-path}
For a general point set, if $T_I$ is a 
non-crossing spanning tree and $T_F$ is a non-crossing path that is monotone in some direction, then $T_I$ can be reconfigured to $T_F$ in 
at most $1.5n -2 - h$ 
flips, where $h=|T_I\cap T_F|$ is the number of happy edges. 
Furthermore, there is a lower bound of $1.5n - 5$ flips, even for points in convex position and if one tree is a monotone path. 
\end{theorem}

\begin{proof} Rotate the plane so that $T_F$ is $y$-monotone.  Note that $T_F$ is then 
both an upward and a downward tree.  We thus have the flexibility to turn $T_I$ into either an upward or a downward tree in the first phase of the algorithm. 
By Corollary~\ref{cor:up-or-down}, 
$T_I$ can be turned into an upward or downward tree $T'_I$ in at most $0.5n-1$ flips. Furthermore, since $T_F$ is a $y$-monotone path, any edge in 
$T_I\cap T_F$ has the form $v_iv_{i+1}$ for some $1\leq i<n$,  and thus, by
Corollary~\ref{cor:up-or-down},
these edges do not flip, which implies that they are still in $T'_I \cap T_F$, so $|T'_I\cap T_F|\geq h$. 
\remove{By Lemma~\ref{lem:tree-up-down} 
\anna{I'm removing that lemma (it turns out that we ONLY use it here).  Let's just make the two-line argument right here.} it takes a total of at most $n-2$ flips to turn $T_1$ into an upward tree and into a downward tree.  Therefore $T_1$ can be turned into an upward or downward tree $T'_1$ in at most $0.5n-1$ flips.  Any edge in $T_1\cap T_2$ has the form $v_iv_{i+1}$ for some $1\leq i<n$ since $T_2$ is a $y$-monotone path.   \changed{By Lemma~\ref{lem:tree-up-down}} all these edges also exist in $T'_1$, so $|T_1'\cap T_2|\geq h$.}
The second phase of the algorithm uses 
\remove{at most $n$ }
$|T_I'\setminus T_F| \leq n-1-h$
perfect flips to reconfigure $T'_I$ into $T_F$. Hence the total number of flips is at most $1.5n-2-h$. 

For the lower bound, see Lemma~\ref{lem:lowerbound} below.
\end{proof}

\begin{lemma}
\label{lem:lowerbound}
On any set of $n \ge 4$ points in convex position, for $n$ even, there exists a non-crossing spanning tree $T_I$ and a
non-crossing path $T_F$ such that 
reconfiguring $T_I$ to $T_F$ requires at least $1.5n - 5$ 
flips, and this bound is tight. 
\end{lemma}
\begin{proof}
Our construction is depicted in Figure~\ref{fig:lowerbound}. 
Note that the tree $T_I$ is the same as in the lower bound of
$1.5n - 5$ proved by Hernando et al.~\cite{hernando1999geometric}, but their tree $T_F$ was not a path.

\begin{figure}[htp]
    \centering
    \includegraphics[page=1,width=\textwidth]{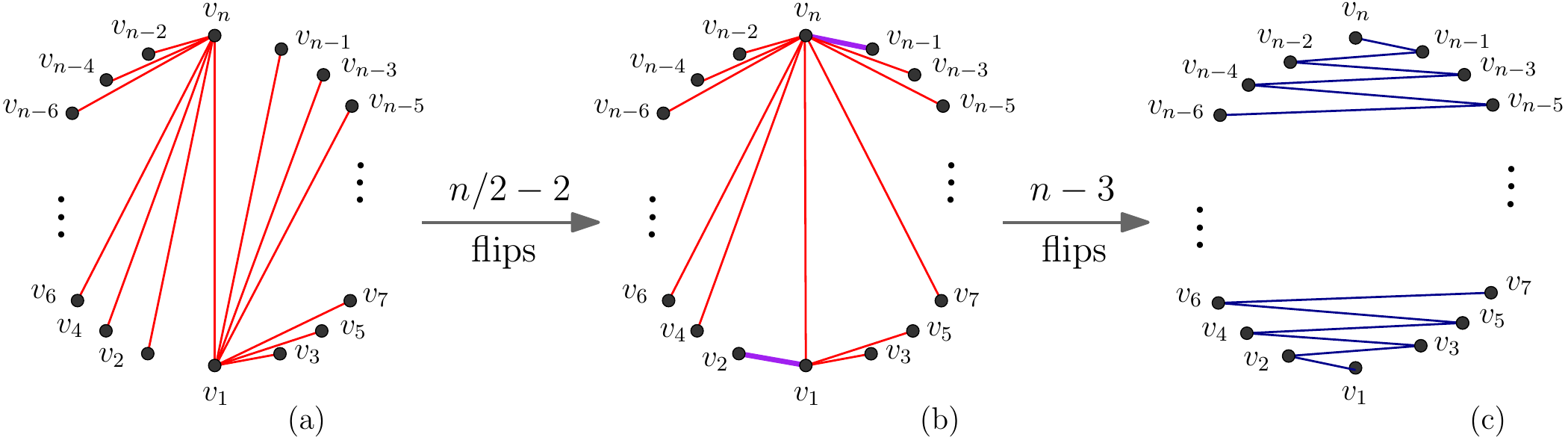}
    \caption{Tight reconfiguration bound: (a) initial tree $T_I$ (b) intermediate tree with two happy hull edges $v_1v_2$ and $v_{n-1}v_n$ (c) final $y$-monotone path $T_F$.}
    \label{fig:lowerbound}
\end{figure}

The construction is as follows (refer to Figure~\ref{fig:lowerbound}(a) and (c)). The points $v_1 \ldots v_n$ (ordered by increasing $y$-coordinate) are placed in convex position on alternate sides of $v_1v_n$. 
$T_I$ contains edges $v_1v_{2i+1}$ and $v_{n}v_{2i}$ for $i=1,\ldots, n/2-1$, and $v_1v_n$.
Note that $v_1$ and $v_n$ have degree $n/2$ each. $T_F$ is a path that connects vertices in order (so it includes edges $v_iv_{i+1}$, for $i = 1, \ldots, n-1$).

Note that every non-hull edge of $T_F$ (in blue) crosses at least $n/2-1$ edges of $T_I$ (in red). 
Indeed, edges of the form  
$v_{2i+1}v_{2i}$ cross exactly $n/2$ edges: 
$n/2-i-1$ edges incident to $v_1$, plus $i+1$ edges incident to $v_n$, minus $1$ because $v_1v_n$ is included in both counts.
Edges of the form 
$v_{2i}v_{2i+1}$ cross one less edge (specifically $v_1v_{2i+1}$).

Thus, any valid reconfiguration from $T_I$ to $T_F$ must flip $n/2-1$ edges of $T_I$ out of the way before the first of the $n-3$ non-hull edges of $T_F$ is added.  After that, we need at least one flip for each of the remaining $n-4$ non-hull edges of $T_F$. Thus the total number of flips is at least $n/2-1 + (n-4)=1.5n-5$. 

We note that our two-phase reconfiguration algorithm uses $1.5n-3$ flips for this instance, but there is a flip sequence of length $1.5n -5$: first flip $v_2v_n$ to $v_1v_2$ to create a happy hull edge, then connect $v_n$ to all $v_i$ for odd $i \ge 7$ by performing the flips $v_1v_i$ to $v_nv_i$ in order of decreasing $i$. The number of flips thus far is $n/2-2$ (note that $v_1v_3$ and $v_1v_5$ stay in place). The resulting tree (shown in Figure~\ref{fig:lowerbound}b) has two happy hull edges. We show that this tree can be reconfigured into $T_F$ using perfect flips only (so the number of flips is $n-3$). Flip $v_1v_n$ to non-hull edge $v_4v_5$ and view the resulting tree as the union of an upward tree rooted at $v_1$ and a downward subtree rooted at $v_n$ (sharing $v_4v_5$). These two subtrees are separated by $v_4v_5$ and therefore can be independently reconfigured into their corresponding subtrees in $T_F$ using perfect flips, as given by Theorem~\ref{thm:trees-up-down}. Thus the total number of flips is $(n/2-2)+ (n-3)=1.5n-5$, proving this bound tight.
\end{proof}

\begin{theorem}
\label{thm:convex-path} For points in convex position, if $T_I$ is a 
non-crossing spanning tree and $T_F$ is a path, then $T_I$ can be reconfigured to $T_F$ in at most $1.5n -2 - h$
flips, where $h=|T_I\cap T_F|$ is the number of happy edges. 
Furthermore, there is a lower bound of $1.5n - 5$ flips.
\end{theorem}

\begin{proof}
When points are in convex position, two edges cross (a geometric property) if and only if their endpoints alternate in the cyclic ordering of points around the convex hull (a combinatorial property). 
This insight allows us to show that the path $T_F$ is ``equivalent to'' a monotone path, which means that we can use the previous Theorem~\ref{tree-and-path}. 
In particular, let the ordering of points in $T_F$ be $v_1, \ldots, v_n$. We claim that the above algorithms can be applied using this ordering in place of the ordering of points by $y$-coordinate.  Thus, a sink in $T_I$ is a point $v_i$ with no edge to a later vertex in the ordering, and etc.  
One could justify this by examining the steps of the algorithms (we relied on geometry only to show that we can add a non-crossing edge ``upward'' from a sink, which becomes easy for points in convex position).  As an alternative, we
make the argument formal by showing how to perturb the points so that $T_F$ becomes a monotone path while preserving the ordering of points around the convex hull---which justifies that the flips for the perturbed points are correct for the original points.

First adjust the points so that they lie on a circle with $v_1$ lowest at $y$-coordinate 1 and $v_n$ highest at $y$-coordinate $n$.  
The convex hull separates into two chains
from $v_1$ to $v_n$.
Observe that $T_F$ visits the points of each chain in order from bottom to top (if $a$ appears before $b$ on one chain but $T_F$ visits $b$ before $a$, then the subpaths from $v_1$ to $b$ and from $a$ to $v_n$ would cross).  Thus, we can place $v_i$ at $y$-coordinate $i$ while preserving the ordering of points around the circle.

We can now apply Theorem~\ref{tree-and-path} to $T_I$ and $T_F$ on the perturbed points. This gives a sequence of at most $1.5n-h-2$ flips to reconfigure $T_I$ into $T_F$ and the flip sequence is still correct on the original points, thus proving the upper bound claimed by the theorem. 
For the lower bound, note that the points in Figure~\ref{fig:lowerbound} are in convex position and $T_F$ is a path. Thus Lemma~\ref{lem:lowerbound} (which employs the example from Figure~\ref{fig:lowerbound}) settles the lower bound claim.
\end{proof}

\section{Improving the Upper Bound for a Convex Point Set}
\label{sec:improved-bound}

In this section we
show that for $n$ points in convex position, reconfiguration between two non-crossing spanning trees can always be done with fewer than $2n$
flips.  

\begin{theorem}
\label{theorem:upper-bound}
There is an algorithm to reconfigure between 
an initial non-crossing spanning tree $T_I$ and a final non-crossing spanning tree $T_F$
on $n$ points in convex position using 
at most 
$2d - \Omega (\log d )$ flips, where $d = |T_I \setminus T_F |$.
\end{theorem}

Before proving the theorem we note 
that previous reconfiguration algorithms do not respect this bound.
Avis and Fukuda~\cite[Section 3.7]{avis1996reverse} proved an upper bound of $2n$ minus a constant by reconfiguring each spanning tree $T_i$ to a star $S$ using $|S \setminus T_i|$ flips.  When $T_i$ is a path, $|S \cap T_i| \le 2$, so $|S \setminus T_i| \ge n-3$ and their method takes at least $2n-6$ flips.  Similarly, the method of flipping both trees to a canonical path around the convex hull takes at least $2n-6$ flips when $T_1$ and $T_2$ are paths with only two edges on the convex hull. 
Although paths behave badly for these canonicalization methods, they are actually easy cases 
as we showed in 
Section~\ref{sec:2phase}. 
As in that section, we do not use a canonical tree to prove Theorem~\ref{theorem:upper-bound}---instead, the flips are tailored to the specific initial and final trees.

Throughout this section, we assume points in convex position.
Consider 
the 
symmetric difference 
$D = (T_I \setminus T_F) \cup (T_F \setminus T_I)$,
so $|D|=2d$.
It is easy to reconfigure $T_I$ to $T_F$ using $2d$ flips---we use $d$ flips to move the edges of  $T_I \setminus T_F$ to the convex hull, giving an intermediate tree $T$, and, from the other end, use $d$ flips to move the edges of $T_F \setminus T_I$ to the same tree $T$.  The plan is to save $\Omega ( \log d )$ of these flips by using that many \emph{perfect flips}
(recall that 
a perfect flip exchanges an edge of $T_I \setminus T_F$ directly with an edge of $T_F \setminus T_I$).
In more detail, the idea is to find an edge $e \in D$ that is crossed by at most (roughly) $d/2$ edges of the other tree. We flip all but one of the crossing edges out of the way to the convex hull, and---if $e$ is chosen carefully---we show that we can perform one flip from $e$ to the last crossing edge, thus providing one perfect flip after at most $d/2$ flips.
Repeating this approach 
$\log d$ times
gives our claimed bound.

We first show how to find an edge $e$ with not too many crossings.  To do this, we define  ``minimal'' edges.
An edge joining points $u$ and $v$ 
determines two subsets of points, those clockwise from $u$ to $v$ and those clockwise from $v$ to $u$ (both sets include $u$ and $v$).  We call these the \defn{sides} of the edge.
An edge is \defn{contained} in a side if both endpoints are in the side.
We call a side \defn{minimal} if it contains no edge of the symmetric difference $D$, and call an edge $e\in D$ \defn{minimal} if at least one of its sides is minimal.
Note that if $D$ is non-empty then it contains at least one minimal edge (possibly a convex hull edge). 
We need the following property of minimal edges.

\begin{figure}[htb]
    \centering
\subfigure[~]{\includegraphics[width=.3\textwidth,trim= 0 0 250 0,clip]{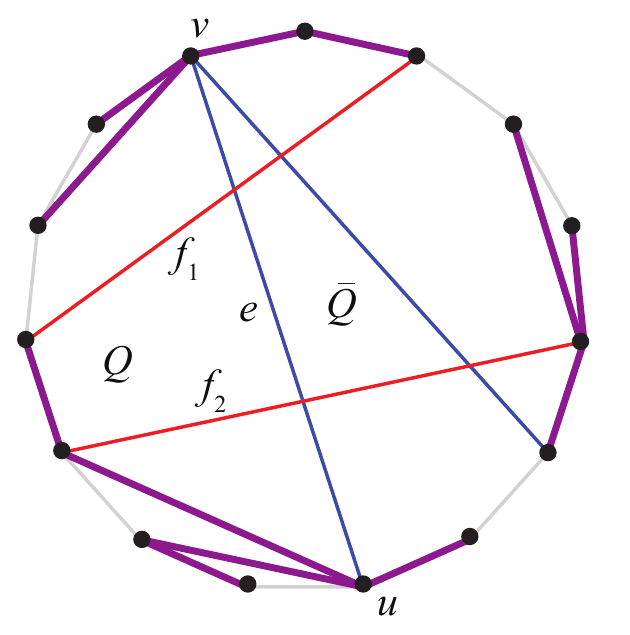}}
\hspace*{10mm}
\subfigure[~]{    \includegraphics[width=.3\textwidth,trim= 250 0 0 0,clip]{figures/ST-reconfig}}
    \caption{(a) Illustration for Claim~\ref{claim:two-components-1} showing $T_I \cap T_F$ in thick purple, $T_I \setminus T_F$ in red, $T_F \setminus T_I$ in blue, and a minimal edge $e \in T_F$.  
    (b) Illustration for Lemma~\ref{lemma:crossing-bound} with $d=6$, showing two minimal edges $a$ and $b$ with $k_a = 5$, $k_b = 4$, and $k_{ab} = 3$.
    }
    \label{fig:minimal-edge}
\end{figure}

\begin{claim}
\label{claim:two-components-1}
Let $e=uv$ be a minimal edge of $D$.  Let $Q$ be a minimal side of $e$, and let $\bar Q$ be the other side.  Suppose $e \in T_F$.  Then
$T_I \cap Q$ consists of exactly two connected components, one containing $u$ and one containing $v$.
\end{claim}
\begin{proof}
The set $T_F \cap Q$ is a non-crossing tree
consisting of edge $e$ and two subtrees $T_u$ containing $u$ and $T_v$ containing $v$.
Since $e$ is minimal, there are no edges of $D$ in $Q$ except for $e$ itself.
This means that
$T_I \cap Q$  
consists of  $T_u$ and $T_v$,
and since $e \notin T_I$, these two components of $T_I \cap Q$ are disconnected in $Q$.
See Figure~\ref{fig:minimal-edge}(a).
\end{proof}

Our algorithm will operate on a minimal edge $e$.  To guarantee the savings in flips, we need a minimal edge with not too many crossings.  

\begin{lemma}
\label{lemma:crossing-bound} If $D$ is non-empty, then there is a minimal edge with at most 
$\lfloor (d+3)/2 \rfloor$ crossings.
\end{lemma}
\begin{proof}
Clearly the lemma holds if some minimal edge is not crossed at all (e.g., a convex hull edge in $D$), so we assume that all minimal edges have a crossing.
Let $a$ be a minimal edge of $D$.
Suppose $a \in T_F$. 
Let $A$ be a minimal side of $a$, and let $\bar A$ be the other side. 
Our plan is to find a second minimal edge $b \in T_F$ such that $b$ is inside $\bar A$ and $b$ has 
a minimal side $B$ that is contained in $\bar A$.  We will then argue that $a$ or $b$ satisfies the lemma.
See Figure~\ref{fig:minimal-edge}(b).

If $\bar A$ is minimal, then set $b := a$ and $B = \bar A$. 
Otherwise, 
let $b$ be an edge of $T_F \setminus T_I$ in $\bar A$ whose $B$ side (the side in $\bar A$) 
contains no other edge of $T_F \setminus T_I$. 
Note that $b$ exists, and that all the edges of $T_F$ in $B$ (except $b$) lie in $T_I$.  
If $b$ is not minimal, then $B$ contains
a minimal edge $c$ (which must then be in $T_I \setminus T_F$), and $c$ is not crossed by any edge of $T_F$, 
because 
such an edge would either have to cross $b$, which is impossible since $b \in T_F$, or lie in $B$, which is impossible since 
all the edges of $T_F$ in $B \setminus \{b\}$ 
are in $T_F \cap T_I$. 
But we assumed that all minimal edges are crossed, so $c$ cannot exist, and so $b$ must be minimal.

Let $k_a$ be the number of edges 
of $T_I$ crossing $a$, let $k_b$ be the number of edges of $T_I$ crossing $b$, and let $k_{ab}$ be the number of edges of $T_I$ crossing both $a$ and $b$. 
Observe that $k_a + k_b \le d + k_{ab}$.

We claim that $k_{ab} \le 3$.  Then $k_a + k_b \le d+3$ so $\min\{k_a,k_b\} \le (d+3)/2$, which 
will complete
the proof since the number of crossings is an integer. 
By Claim~\ref{claim:two-components-1}, $T_I \cap A$ has two connected components and $T_I \cap B$ has two connected components. 
Now 4 connected components in a tree can have at most three edges joining them, which
implies that $k_{ab} \le 3$.
(Note that this argument is correct even for $a=b$, though we get a sharper bound since $k_a = k_b = k_{ab} \le 3$.)
\end{proof}

\paragraph*{Algorithm.}

Choose a minimal edge $e$ with $k$ crossings, where $0 \le k \le (d+3)/2$ (as guaranteed by Lemma~\ref{lemma:crossing-bound}). 
Suppose $e \in T_F$, so the crossing edges belong to $T_I$. The case where $e \in T_I$ is symmetric. In either case the plan is to 
perform some flips on $T_I$ and some on $T_F$ to reduce the difference $d$
by $k$ (or by 1, if $k=0$) and apply the algorithm recursively to the resulting instance.  Note that the algorithm constructs a flip sequence by adding flips at both ends of the sequence.  

If $k=0$ then add $e$ to $T_I$.  This creates a cycle, and the cycle must have an edge $f$ in $T_I \setminus T_F$. Remove $f$.  This produces a new tree $T_I$.  We have performed one perfect flip and reduced $d$ by 1.
Now recurse.

Next suppose $k \ge 1$.
Let $e = uv$.  
Let $Q$ be the  minimal side of $e$ and let $\bar Q$ be the other side (both sets include $u$ and $v$).
Let $f_1, \ldots, f_k$ be the edges that cross $e$.  We will flip all but the last crossing edge to the convex hull.
\changed{}
For $i=1, \ldots k-1$ we flip $f_i$ as follows.

\begin{enumerate}
\item Perform a flip in $T_I$ by removing $f_i$ and adding a convex hull edge $g$ that lies in $\bar Q$. (The existence of $g$ is proved below.)

\item If $g \in T_F$ then this was a perfect flip and we have performed one perfect flip and reduced $d$ by 1. 

\item Otherwise (if $g \notin T_F$), perform a flip in $T_F$ by adding $g$ and removing an edge $h \in T_F \setminus T_I$,
that lies in $\bar Q$ and is not equal to $e$.
(The existence of $h$ is proved below.)

\end{enumerate}

At this point, only $f_k$ crosses $e$.
Perform one flip in $T_I$ to remove $f_k$ and add $e$. 
(Correctness proved below.)
Now, apply the algorithm recursively to the resulting $T_I, T_F$.

This completes the description of the algorithm.

\paragraph*{Correctness.}
We must prove that $g$ and $h$ exist and that the final flip is valid.

First note that $e$ remains a minimal edge 
after each flip performed inside the loop
because we never add or remove edges inside $Q$.
We need one more invariant of the loop.

\begin{claim}
\label{claim:disconnected-Q-bar}
Throughout the loop $u$ and $v$ are disconnected in $T_I \cap {\bar Q}$. 
\end{claim}
\begin{proof}
Suppose there  is a path $\pi$ from $u$ to $v$ in $T_I \cap {\bar Q}$. 
Now consider the edge $f_k$ which crosses $e$, say from $x \in Q$ to $y \in {\bar Q}$.  Since $f_k$ cannot cross $\pi$, we have $y \in \pi$.   
By Claim~\ref{claim:two-components-1}, $T_I \cap Q$ consists of two components, one containing $u$ and one containing $v$.  Suppose, without loss of generality, that $x$ lies in the component containing $u$. Then there is a path from $x$ to $u$ in $T_I \cap Q$ and a path from $u$ to $y$ in $T_I \cap {\bar Q}$, and these paths together with $f_k$ make a cycle in $T_I$, a contradiction. 
\end{proof}

First we prove that $g$ exists in Step 1. Removing $f_i$ from $T_I$ disconnects $T_I$ into two pieces. There are two convex hull edges that connect the two pieces.  By Claim~\ref{claim:two-components-1}, $T_I \cap Q$ consists of two connected components, one containing $u$ and one containing $v$.  Thus at most one of the convex hull edges lies in $Q$, so at least one lies in $\bar Q$.

Next we prove that $h$ exists in Step 3. 
Adding $g$ to $T_F$ creates a cycle $\gamma$ in $T_F$ and this cycle must lie in $\bar Q$ (because $e \in T_F$) and must contain at least one edge of $T_F \setminus T_I$ (because $T_I$ does not contain a cycle). If $e$ were the only edge of $T_F \setminus T_I$ in $\gamma$, then $u$ and $v$ would be joined by a path in $T_I \cap {\bar Q}$, contradicting Claim~\ref{claim:disconnected-Q-bar}.  Thus $h$ exists.

Finally, we prove that the last flip in $T_I$ (to remove $f_k$ and add $e$) is valid. Removing $f_k$ leaves $u$ and $v$ disconnected in $Q$ by Claim~\ref{claim:two-components-1} and disconnected in $\bar Q$ by Claim~\ref{claim:disconnected-Q-bar}.  Adding $e$ reconnects them, and yields a non-crossing spanning tree.

\paragraph*{Analysis.}
We now prove that the algorithm uses at most the claimed number of flips. 

\begin{observation}
\label{obs:minimal-edge-analysis}
In each recursive call: 
if $k=0$, then the algorithm performs one perfect flip and reduces $d$ by 1; 
and if $k>0$, then the algorithm performs at most $2k-1$ flips (one for $f_k$ and at most 2 for each other $f_i$) and reduces $d$ by $k$ (in each loop iteration, $g$ joins the happy set $T_I \cap T_F$ and in the final step $e$ joins the happy set).  
\end{observation}

\begin{lemma}
The number of flips performed by the algorithm is at most $2d- \lfloor \log(d{+}3) \rfloor +1$.
\end{lemma}
\begin{proof}
We prove this by induction on $d$.   In the base case $d=0$ we perform $0=2d-\lfloor \log(d{+}3)\rfloor +1$ flips.

Now assume $d\geq 1$, 
and consider what happens in the first recursive call of the algorithm, see  Observation~\ref{obs:minimal-edge-analysis}. 
If the algorithm chooses an edge with $k=0$ crossings, then the algorithm performs one perfect flip.
The resulting instance has a difference set of size $d' = d - 1$ and induction applies, so in 
the total number of flips we perform is at most
\begin{align*}
1 + 2d'-\lfloor \log(d'{+}3)\rfloor + 1
= 2d-\lfloor \log(d{+}2)\rfloor \leq
2d-\lfloor \log(d{+}3)\rfloor +1,
\end{align*}
which proves the result in this case.

Now suppose that the algorithm chooses an edge with $k \ge 1$ crossings, where $k \le \lfloor (d{+}3)/2 \rfloor$.
The algorithm performs at most $2k-1$ flips and the resulting instance has a difference set
of size $d' = d - k$ and therefore 
$d'{+}3  \ge d{+}3 - \lfloor (d{+}3)/2 \rfloor  = \lceil (d{+}3)/2 \rceil \geq (d{+}3)/2$.   
By induction, the total number
of flips that we perform is hence at most
\remove{
\begin{align*}
2k-1 + 2d'-\lfloor \log(d'{+}3)\rfloor + 1
\mchanged{\leq}~2d-\lfloor \log((d{+}3)/2)\rfloor \leq 2d - \lfloor \log (d{+}3)\rfloor +1
\end{align*}
}

\begin{align*}
(2k-1) + (2d'-\lfloor \log(d'{+}3)\rfloor + 1)
& \leq (2k-1) + (2(d-k) - \lfloor \log((d{+}3)/2)\rfloor + 1)\\
& \leq~2d-\lfloor \log((d{+}3)/2)\rfloor\\ 
& \leq 2d - \lfloor \log (d{+}3)\rfloor +1
\end{align*}

as desired. 
\end{proof}

This completes the proof of Theorem~\ref{theorem:upper-bound}.
\bigskip

\section{The Happy Edge Conjecture}
\label{sec:happy-edges}

In this section we make some conjectures and prove some preliminary results in attempts to characterize \emph{which} edges need to be flipped in minimum flip sequences for non-crossing spanning trees.

Recall that
an edge $e$ is \emph{happy} if
$e$ lies in $T_I \cap T_F$.  
We make the following 
conjecture for points in convex position.  In fact, we do not have a counterexample even for general point sets, though our guess is that the conjecture fails in the general case.

\begin{conjecture} {\bf[Happy Edge Conjecture for Convex Point Sets]}
\label{conj:happy-edge}
For any point set $P$ in convex position and any two non-crossing spanning trees $T_I$ and $T_F$ of $P$, there is a minimum flip sequence from $T_I$ to $T_F$ such that 
no happy edge is flipped during the sequence.
\end{conjecture}

In this section we first prove this conjecture for the case of happy edges on the convex hull. Then in Section~\ref{sec:parking-edges} we make some stronger conjectures about which extra edges (outside $T_I$ and $T_F$) might be needed in minimum flip sequences.  In Section~\ref{sec:greedy-perfect-flips} we show that even if no extra edges are needed, it may be tricky to find a minimum flip sequence---or, at least, a greedy approach fails.
Finally, in Section~\ref{sec:edge-slides} we prove that the Happy Edge Conjecture fails if we restrict the flips to ``slides'' where one endpoint of the flipped edge is fixed and the other endpoint moves along an adjacent tree edge. 

If the Happy Edge Conjecture is false then a minimum flip sequence might need to remove an  edge and later add it back. 
We are able to prove something about such ``remove-add'' subsequences, even for general point sets:

\begin{proposition}
\label{prop:remove-add}
Consider any point set $P$ and any two non-crossing spanning trees $T_I$ and $T_F$ on $P$ and any minimum flip sequence from $T_I$ to $T_F$.  If some edge $e$ is removed and later added back, then 
some flip during that subsequence must 
add an edge $f$ that crosses $e$.
\end{proposition}

Before proving this Proposition, we note the implication that the Happy Edge Conjecture is true for convex hull edges:

\begin{corollary}
Conjecture~\ref{conj:happy-edge} is true for happy edges on the convex hull. 
Furthermore, \emph{every} minimum flip sequence 
keeps the happy convex hull edges throughout the sequence.
\end{corollary}
\begin{proof}
Let $e$ be a happy convex hull edge.
Suppose for a contradiction that there is a minimum flip sequence in which $e$ is removed.  Note that $e$ must be added back, since it is in $T_F$. By Proposition~\ref{prop:remove-add}, the flip sequence must use an edge $f$ that crosses $e$.  But that is impossible because $e$ is a convex hull edge so nothing crosses it. 
\end{proof}

\begin{proof}[Proof of Proposition~\ref{prop:remove-add}]
Consider a flip sequence from $T_I$ to $T_F$ and suppose that an edge $e$ is removed and later added back, and that no edge crossing $e$ is added during that subsequence.  We will make a shorter flip sequence.  The argument is similar to the ``normalization'' technique used by Sleator et al.~\cite{sleator1988} to prove the happy edge result for flips in triangulations of a convex point set.

Let $T_0, \ldots, T_k$ be the trees in the subsequence, where $T_0$ and $T_k$ contain $e$, but none of the intervening trees do.  Suppose that none of the trees $T_i$ contains an edge that crosses $e$.
We will construct a shorter flip sequence from $T_0$ to $T_k$.  
For each $i$, $0 \le i \le k$ consider adding $e$ to $T_i$.  For $i \ne 0,k$, this creates a cycle $\gamma_i$. 
Let $f_i$ be the first edge of $\gamma_i$ 
that is removed during the flip sequence from $T_i$ to $T_k$.  Note that $f_i$ exists since $T_k$ contains $e$, so it cannot contain all of $\gamma_i$.  
Define $N_i = T_i \cup \{e\} \setminus \{f_i\}$ for $1 \le i \le k-1$, and define $N_0 := T_0$.
Observe that $N_i$ is a spanning tree, and is non-crossing because no edge of $T_i$ crosses $e$ by hypothesis.
Furthermore,  
$N_{k-1} = T_k$
because the flip from $T_{k-1}$ to $T_k$ is exactly the same as the flip from $T_{k-1}$ to $N_{k-1}$.  

We claim that $N_0, \ldots, N_{k-1}$ is a flip sequence.  This will complete the proof, since it is a shorter flip sequence from $T_0$ to $T_k$.

Consider $N_i$ and $N_{i+1}$.  Suppose that the flip from $T_i$ to $T_{i+1}$ adds $g$ and removes $h$.

\[
\begin{tikzcd}[row sep=large, column sep=large]
T_i \arrow{r}{+g,-h} \arrow[swap]{d}{+e,-f_i} & T_{i+1} \arrow{d}{+e,-f_{i+1}} \\
N_i \arrow{r}{+g,-?} & N_{i+1}
\end{tikzcd}
\]

Recall that $\gamma_i$ 
is the cycle containing $e$ in $T_i \cup e$. 
If $h$ belongs to $\gamma_i$ 
then $f_i = h$, and then to get from $N_i$ to $N_{i+1}$ we add $g$ and remove $f_{i+1}$.
Next, suppose that $h$ does not belong to $\gamma_i$. 
Then the cycle $\gamma_i$ 
still exists in $T_{i+1}$.  Now, $\gamma_{i+1}$ 
is the unique cycle in $T_{i+1} \cup e$.  Thus $\gamma_{i+1} = \gamma_i$. 
Furthermore, $f_{i+1}$ is by definition the first edge removed from $\gamma_{i+1}$ 
in the flip sequence from $T_{i+1}$ to $T_k$. Thus $f_{i+1} = f_i$. Therefore, to get from $N_i$ to $N_{i+1}$ we add $g$ and remove $h$.

This shows that a single flip changes $N_i$ to $N_{i+1}$, which completes the proof.
\end{proof}

Note that the proof of Proposition~\ref{prop:remove-add} produces a strictly shorter flip sequence.  But to prove the Happy Edge Conjecture (Conjecture~\ref{conj:happy-edge}) it would suffice to produce a flip sequence of the same length.
One possible approach is to consider how 
remove-add pairs and add-remove pairs 
interleave in a flip sequence. 
Proposition~\ref{prop:remove-add} shows that 
a remove-add pair for edge $e$ must contain an add-remove pair for $f$ inside it.
We may need to understand
how the order of flips can be rearranged in a flip sequence.  
Such flip order rearrangements 
are at the heart of results on triangulation flips, both for convex point sets~\cite{sleator1988,pournin2014diameter} and for general point sets~\cite{kanj2017computing}.

\subsection{Extra edges used in flip sequences}
\label{sec:parking-edges}
Any flip sequence from $T_I$ to $T_F$ must involve flips that remove edges of $T_I \setminus T_F$ and flips that add edges of $T_F \setminus T_I$.  
Recall that in a perfect flip sequence, these are the only moves and they pair up perfectly, so the number of flips is $| T_I \setminus T_F |$.    
Theorem~\ref{thm:trees-up-down} gives one situation where a perfect flip sequence is possible, but typically (e.g.,~in the example of Figure~\ref{fig:lowerbound})
we must add edges not in $T_F$, and later remove them.  
More formally, an edge outside $T_I \cup T_F$ that is used in a flip sequence is called a \defn{parking edge},
with the idea that we ``park'' edges there temporarily.  

We make two further successively stronger conjectures.  They may not hold, but disproving them would give more insight.

\begin{conjecture}
\label{conj:non-crossing}
For any point set $P$ in convex position and any two non-crossing spanning trees $T_I$ and $T_F$ of $P$ there is a minimum flip sequence from $T_I$ to $T_F$ that never 
uses a parking edge that crosses an edge of $T_F$.
\end{conjecture}

\begin{conjecture}
\label{conj:CH-parking}
For a point set $P$ in convex position and any two non-crossing spanning trees $T_I$ and $T_F$ on $P$ there is a minimum flip sequence from $T_I$ to $T_F$ that only uses 
parking edges from the convex hull.
\end{conjecture}

Our experiments verify Conjecture~\ref{conj:CH-parking} for $n \le 10$ points, (see Observation~\ref{obs:happy-edge-experiments}).
We note that Conjecture~\ref{conj:CH-parking} cannot hold for general point sets (there just aren't enough convex hull edges).  However, we do not know if Conjecture~\ref{conj:non-crossing} fails for general point sets.

\begin{claim}
Conjecture~\ref{conj:CH-parking} $\Longrightarrow$ Conjecture~\ref{conj:non-crossing} $\Longrightarrow$ 
Conjecture~\ref{conj:happy-edge}.
\end{claim}
\begin{proof}
The first implication is clear. 
For the second implication we use  Proposition~\ref{prop:remove-add}. Consider the minimum flip sequence promised by Conjecture~\ref{conj:non-crossing}.  If there is a happy edge $e \in T_I \cap T_F$ that is removed during this flip sequence, then by Proposition~\ref{prop:remove-add}, the flip sequence must add an edge $f$ that crosses $e$.  But then $f$ is a parking edge that crosses an edge of $T_F$, a contradiction.
\end{proof}

\subsection{Finding a perfect flip sequence---greedy fails}
\label{sec:greedy-perfect-flips}

It is an open question whether there is a polynomial time algorithm to find [the length of] a minimum flip sequence between two given non-crossing spanning trees $T_I$ and $T_F$.
A more limited goal is testing whether there is a flip sequence of length $|T_I \setminus T_F|$---i.e., whether there is a perfect flip sequence.  This is also open.

In Figure~\ref{fig:Therese-example} we give an example to show that 
a greedy approach to finding a perfect flip sequence may fail.  
In this example there is a perfect flip sequence but a poor choice of perfect flips leads to a dead-end configuration where no further perfect flips are possible. Note that choosing perfect flips involves pairing edges of $T_I \setminus T_F$ with edges of $T_F \setminus T_I$ as well as ordering the pairs.

\begin{figure}[ht]
    \centering
    \includegraphics[width=.9\textwidth]{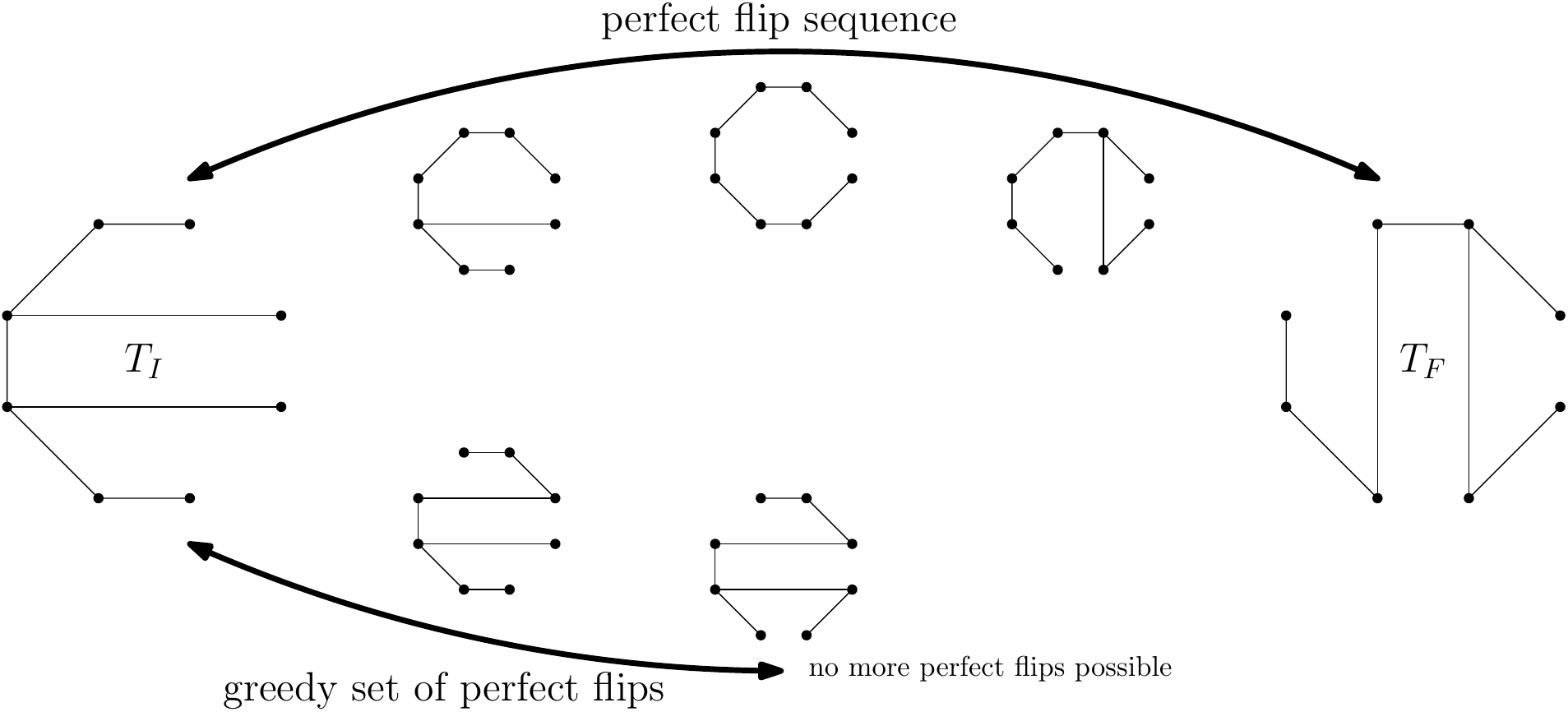}
    \caption{Even if a perfect flip sequence exists, we do not necessarily find it by greedily executing perfect flips.}

    \label{fig:Therese-example}
\end{figure}

\subsection{The Happy Edge Conjecture fails for edge slides}
\label{sec:edge-slides}

Researchers have examined various restricted types of flips for non-crossing spanning trees, see~\cite{nichols2020transition}.
An \defn{edge slide} is the most
restricted 
flip operation possible: it keeps one endpoint of the flipped edge fixed and moves the
other one along an adjacent tree edge without intersecting any of the other edges or vertices of the tree. In other words, the edge that is  removed, the edge that is inserted, 
and the edge along which the slide takes place form an empty triangle.  Aichholzer et al.~\cite{aichholzerreinhardt2007155} proved that for any set $P$ of $n$ points in the plane it is possible to transform between any two non-crossing spanning trees of $P$  
using $O(n^2)$ edge slides. The authors also give an example to show that $\Omega(n^2)$ slides might be required even if the two spanning trees differ in only two edges. This example already implies that for point sets in general position the Happy Edge Conjecture fails for edge slides. We will show that this is also the case for points in convex position.

\begin{figure}[htb]
    \centering
    \includegraphics[width=.9\textwidth]{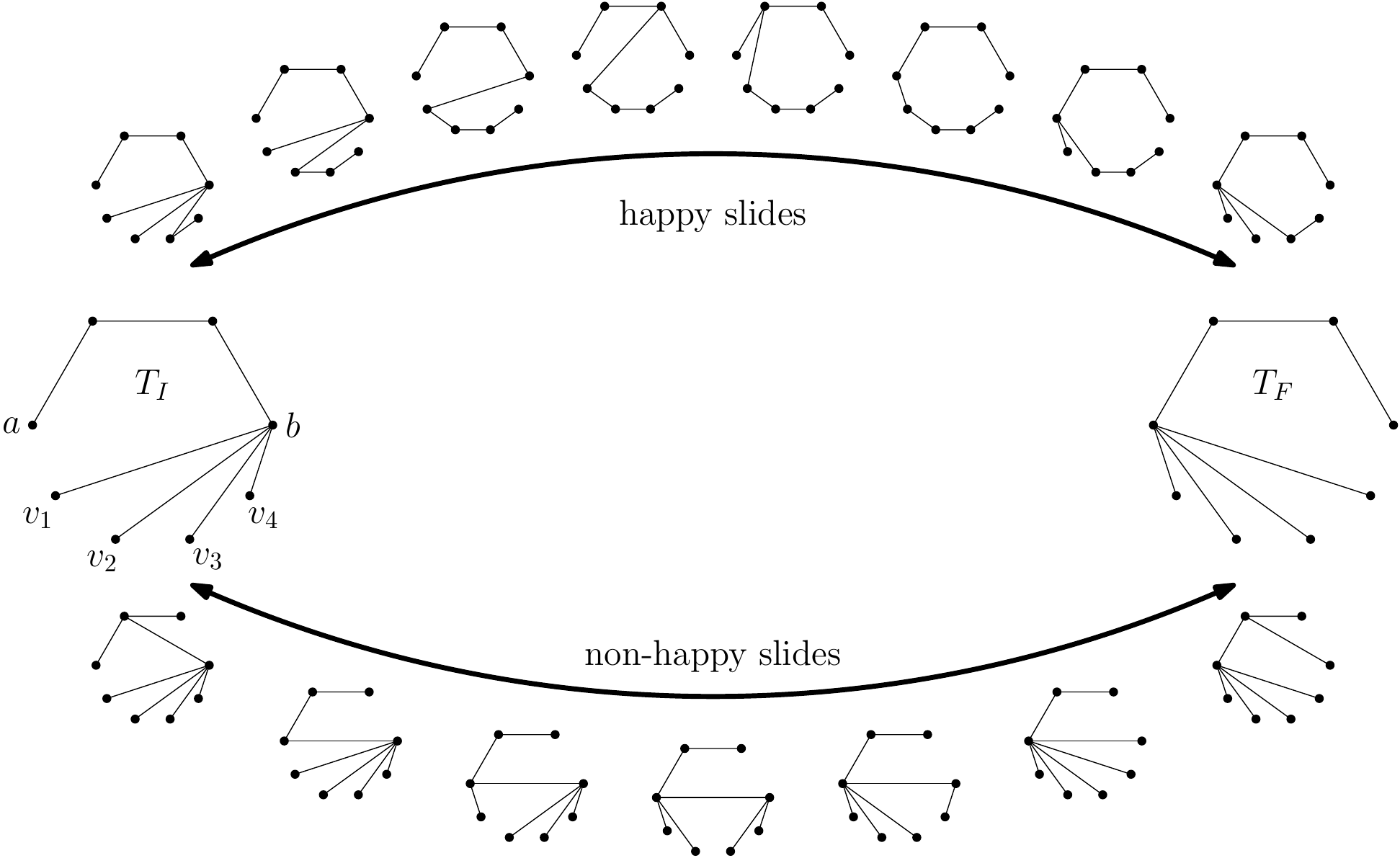}
    \caption{When flips are restricted to slide along an existing edge the Happy Edge Conjecture fails even for sets of points in convex position: Flipping from tree $T_I$ to tree $T_F$ needs 9 flips when respecting happy edges (top), but can be done with 8 flips (bottom) when using an edge (the edge upward from vertex $b$) which is common to both the start and target tree.}
    \label{fig:sliding_nonhappy}
\end{figure}

Figure~\ref{fig:sliding_nonhappy} shows an example of two plane spanning trees $T_I$ and $T_F$ on $8$ points in convex position which can be transformed into each other with 8 slides, shown at the bottom of the figure. To obtain this short sequence we temporarily use an edge which is common to both trees to connect the two vertices $a$ and $b$. Thus this sequence contains a non-happy slide operation, that is, an edge that is common to both, $T_I$ and $T_F$ is moved. When flipping from tree $T_I$ to tree $T_F$ by using only happy slide operations there are some useful observations. First, there can not be an edge directly connecting $a$ and $b$, as this would cause a cycle. This implies that any edge which connects a vertex $v_i$, $1 \leq i \leq 4$, with $b$ needs at least two slides to connect $a$ to some (possible different) vertex $v_j$. Moreover, the first of these edges that gets connected to $a$ needs at least three slides, as at the beginning this is the shortest path connecting $b$ to $a$. Thus in total we need at least $3+2+2+2=9$ happy slide operations. Figure~\ref{fig:sliding_nonhappy}(top) shows such a sequence.  It is not hard to see that this example can be generalized to larger $n$ and implies that the Happy Edge Conjecture fails for points in convex position.

\section{Exhaustive search over small point sets in convex position}
\label{sec:exhaustive}

For small point sets in convex position we investigated the minimum flip distance between non-crossing spanning trees by exhaustive computer search. Table~\ref{tab:exhaustive} summarizes the obtained results. For $n=3,\ldots ,12$ we give the number of non-crossing spanning trees (which is sequence A001764 in the On-Line Encyclopedia of Integer Sequences, \url{https://oeis.org}) and the number of reconfiguration  
steps
between them. Moreover, we computed the maximum reconfiguration distance between two trees 
(the diameter of the reconfiguration graph)
as well as the radius of the reconfiguration graph.
We provide the same information for the special case when the trees are non-crossing spanning paths. Note that in this case the intermediate graphs can still be non-crossing spanning trees. For the case where all intermediate graphs are also non-crossing spanning paths the diameter of the reconfiguration graph for points in convex position is known to be $2n-6$ for $n \ge 5$~\cite{akl2007planar,chang2009diameter}.

\begin{table}[ht]
\begin{tabular}{ r || r | r || r | r || r | r | r }
$n$ & number of   & number of  & max flip & flip   & number of   & path max. & path flip \\
    & plane trees & flip edges & distance & radius & plane paths & flip dist.& radius    \\
 \hline
 3 &         3 &           3 &  1 &  1 &     3 &  1 &  1 \\
 4 &        12 &          32 &  3 &  2 &     8 &  3 &  2 \\
 5 &        55 &         260 &  4 &  3 &    20 &  4 &  3 \\
 6 &       273 &       1 920 &  5 &  4 &    48 &  5 &  4 \\
 7 &     1 428 &      13 566 &  6 &  5 &   112 &  6 &  5 \\
 8 &     7 752 &      93 632 &  8 &  6 &   256 &  7 &  6 \\
 9 &    43 263 &     637 560 &  9 &  7 &   576 &  8 &  7 \\
10 &   246 675 &   4 305 600 & 11 &  8 & 1 280 & 10 &  8 \\
11 & 1 430 715 &  28 925 325 & 12 &  9 & 2 816 & 11 &  9 \\
12 & 8 414 640 & 193 666 176 & 14 & 10 & 6 144 & 13 & 10 \\

\end{tabular}
\caption{For a set of $n$ points in convex position this table gives the size of the reconfiguration graph (the number of non-crossing (``plane'') spanning trees and the number of reconfiguration edges) the maximum reconfiguration distance and radius. For the special case of non-crossing (``plane'') spanning paths also number, distance, and radius are given.}
\label{tab:exhaustive}
\end{table}

\paragraph{Results.}
Our computations show that
for small sets in convex position 
the radius of each reconfiguration graph is strictly larger than half 
the diameter.
More precisely, for $6 \leq n \leq 12$ the diameter is $\lfloor 1.5n-4 \rfloor$  and the radius is $n-2$ which would give an upper bound for the diameter of only $2n-4$.
But this might be an artefact of small numbers: compare for example the result of Sleator, Tarjan, and Thurston which give the upper bound of $2n-10$ for the rotation distance of binary trees which is tight only for $n \geq 13$~\cite{pournin2014diameter,sleator1988}. That the radius seems not to be suitable for obtaining a tight bound for the diameter also supports our way of bounding the diameter of the reconfiguration graph by not using a central canonical tree.

In addition to the results shown in the table, we checked, for $n \le 10$, \emph{which} edges are exchanged, in order to test the Happy Edge Conjecture (Conjecture~\ref{conj:happy-edge}) and whether only parking edges on the convex hull are used (Conjecture~\ref{conj:CH-parking}).

\begin{observation}
\label{obs:happy-edge-experiments}
For $n \leq 10$ points in convex position (1) the Happy Edge Conjecture is true, and (2) there are always minimum flip sequences that only use parking edges on the convex hull. 
\end{observation}

\paragraph{Methods.}
These computations are rather time consuming, as in principle for any pair of non-crossing spanning trees (paths) the flip distance has to be computed. For an unweighted and undirected graph $G$ with $n'$ nodes (non-crossing spanning trees in our case) and $m'$ edges (edge exchanges in our case) the standard algorithm to compute the diameter of $G$ is to apply breath first search (BFS) for each node of $G$. The time requirement for this simple solution is $O(n'm')$. There exist several algorithms which achieve a better running time for graphs of real world applications, see e.g.,~\cite{CRESCENZI201384}, but in the worst case they still need $O(n'm')$ time. The basic idea behind these approaches is to compute the eccentricity $e(v)$ of a node $v \in G$ (which is the radius as seen from this node $v$), and compare this with the largest distance $d$ between two nodes found so far. If $e(v)=d/2$ we know that the diameter of the graph is $d$ and the algorithm terminates. The difference between the various algorithms is how the nodes for computing the eccentricity and the lower bound for the diameter are chosen and the performance of the approaches are usually tested by applying them to a set of examples.

However, it turned out that by the structure of our reconfiguration graphs these approaches do not perform better than the simple
textbook
solution. 
Because the radius of the reconfiguration graph is strictly larger than half the diameter in our test cases, no algorithmic shortcut is possible.
\remove{Actually, for small sets in convex position we verified that the radius of each reconfiguration graph is strictly larger than half of the diameter of it (c.f. Table~\ref{tab:exhaustive}), so no algorithmic shortcut is possible.
At least for small sets it seems that the radius is $n-2$ which would give an upper bound for the diameter of only $2n-4$.
But this might be an artefact of small numbers: compare for example the result of Sleator, Tarjan, and Thurston which give the upper bound of $2n-10$ for the rotation distance of binary trees which is tight only for $n \geq 13$~\cite{pournin2013combinatorial,sleator1988}. That the radius seems not to be suitable for obtaining a reasonable bound for the diameter also supports our way of bounding the diameter of the reconfiguration graph by not using a central canonical tree.
}

To still be able to compute the diameter of the rather large graphs (for $n=12$ the reconfiguration graph has $8~414~650$ nodes and $193~666~176$ edges) we make use of the inherent symmetries of our graphs. For every tree $T$ we can cyclically shift the labels of the vertices (by $1$ to $n-1$ steps) and/or mirror the tree to obtain another non-crossing spanning tree $T'$ of the convex point set. All trees that can be obtained this way can be grouped together. While every tree is needed in the reconfiguration graph to correctly compute shortest reconfiguration distances, by symmetry a call of BFS for any tree from the same group will result in the same eccentricity. It is thus sufficient to call BFS only for one tree of each group. For $n$ points this reduces the number of calls by almost a factor of $2n$, as the size of the group can be up to $2n$ (some trees are self-symmetric in different ways, thus some groups have a cardinality less than $2n$). 

For our experiments on \emph{which} edges are exchanged (for Observation~\ref{obs:happy-edge-experiments}),
the computations get even more involved. The reason is that these properties of edges are defined by the initial and final tree. So it can happen that a short sub-path is valid only for some, but not all, pairs of trees where we would like to use it. Moreover, for similar reasons we can not make full use of the above described symmetry.
This is the reason why we have been able to test our conjectures only for sets with up to $n=10$ points.

\section{Conclusions and Open Questions}

We conclude with some open questions:

\begin{enumerate}

\item We gave two algorithms to find flip sequences for non-crossing spanning trees and we bounded the length of the flip sequence.  The algorithms run in polynomial time, but it would be good to optimize the run-times.

\item A main open question is to close the gap between $1.5n$ and $2n$ for the leading term 
of the diameter of the reconfiguration graph of non-crossing spanning trees.

\item A less-explored problem is to find the radius of the reconfiguration graph (in the worst case, as a function of $n$, the number of points). 
Is there a lower bound of $n-c$ on the radius of the reconfiguration graph for some small constant $c$?

\item Prove or disprove the Happy Edge Conjecture.

\item Is the distance problem (to find the minimum flip distance between two non-crossing spanning trees) NP-complete for general point sets?  For convex point sets? 
A first step towards an NP-hardness reduction would be to find instances where the Happy Edge Conjecture fails.

\item An easier question is to test whether there is a perfect flip sequence between two non-crossing spanning trees.  Can that be done in polynomial time, at least for points in convex position?

\item Suppose the Happy Edge Conjecture turns out to be false.  Is the following problem NP-hard?  
Given two trees, is there a minimum flip sequence between them that does not flip happy edges?

\item Suppose we have a minimum flip sequence that does not flip happy edges and does not use parking edges (i.e., the flips only involve edges of the difference set $D = (T_I \setminus T_F) \cup (T_F \setminus T_I)$).  Is it a perfect flip sequence?

\item All the questions above can be asked for the other versions of flips between non-crossing spanning trees (as discussed in Section~\ref{sec:background} and surveyed in~\cite{nichols2020transition}).

\item For the convex case, what if we only care about the cyclic order of points around the convex hull, i.e., we may freely relabel the points so long as we preserve the cyclic order of the labels. This ``cyclic flip distance'' may be less than the standard flip distance.  
For example, two stars rooted at different vertices have cyclic flip distance $0$ but standard flip distance $n-2$. 
\end{enumerate}

\paragraph*{Acknowledgements}
This research was begun in the Fall 2019 Waterloo problem solving session run by Anna Lubiw and Therese Biedl.  For preliminary progress on the Happy Edge Conjecture, we thank the participants, Pavle Bulatovic, Bhargavi Dameracharla, Owen Merkel, Anurag Murty Naredla, and Graeme Stroud.  
Work on the topic continued in the 2020 and 2021 Barbados workshops run by Erik Demaine and we thank the other participants for many valuable discussions.

\bibliographystyle{plainurl}   
\bibliography{flip}

\end{document}

%% file: proof.tex
\subsection{Phase 2: Reconfiguring an upward tree into a downward tree}
\label{sec:phase2}

In this section we show how to reconfigure 
from an initial downward tree $T_I$ to a final upward tree $T_F$ on a general point set using only perfect flips. 
Thus the total number of flips will be $|T_I \setminus T_F|$.
%
%
%
%
The sequence of flips is simple to describe, and it will be obvious that each flip yields a spanning tree.  What needs proving is that each intermediate tree is non-crossing.
To simplify the description of the algorithm, imagine $T_I$ colored red and $T_F$ colored blue. Refer to Figure~\ref{fig:phase2-setup}.
Recall that
$v_1, \ldots, v_n$ is the ordering of the points by increasing $y$-coordinate.
Define $b_i$ to be the (unique) blue edge in $T_F$ going down from $v_i$, $i=2,\dots,n$.
An \defn{unhappy} edge is an edge of $T_I\setminus T_F$, i.e., it is red but not blue.


\begin{figure}[htp]
    \centering
    \includegraphics[page=1,width=0.6\textwidth]{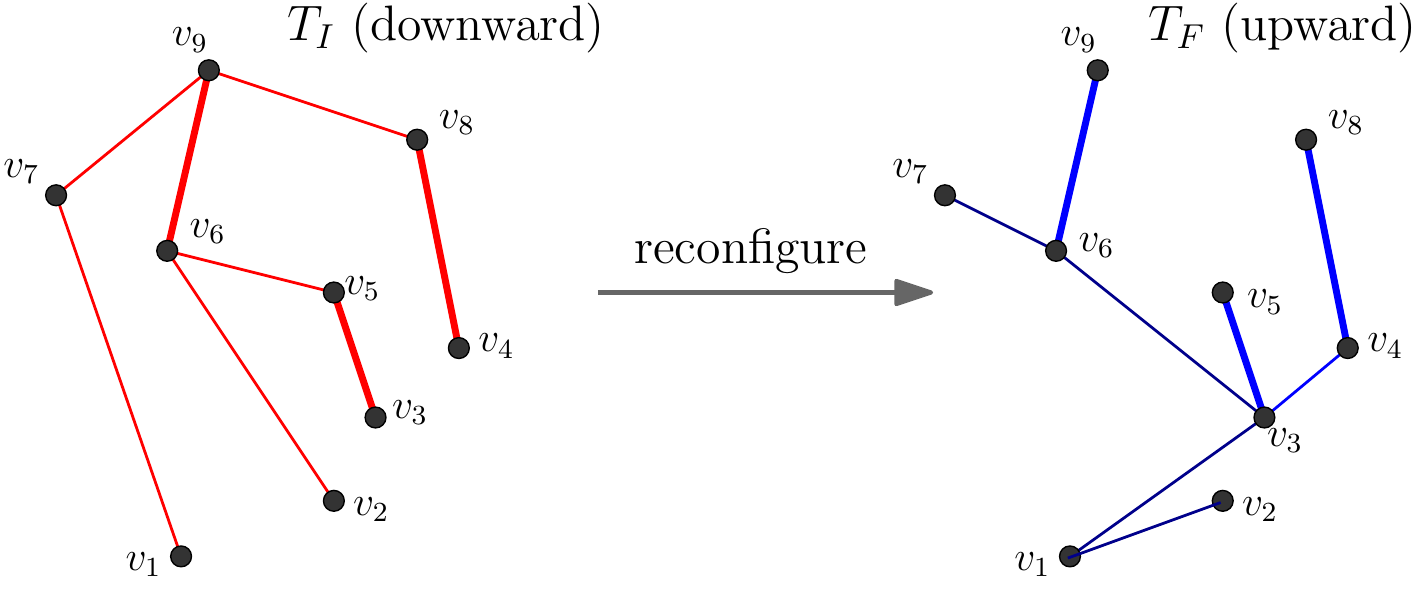}
    \caption{Reconfiguring opposite trees $T_I$ to $T_F$ with happy 
    edges marked in thick lines.
    } 
    \label{fig:phase2-setup}
\end{figure}

\paragraph{Reconfiguration algorithm.}

Let $T_1 = T_I$.  
For $i=2, \ldots, n$ we create a tree $T_i$ that contains $b_2, \ldots, b_i$ and all the happy edges.
If $b_i$ is happy, then $b_i$ is already in the current tree and we simply set $T_{i} := T_{i-1}$.  
Otherwise, consider the cycle formed by adding $b_i$ to $T_{i-1}$ and, in this cycle, 
let $r_i$ be the unhappy red edge with the lowest bottom endpoint.
Note that $r_i$ exists, otherwise all edges in the cycle would be blue.
Set $T_i := T_{i-1} \cup \{b_i\} \setminus \{r_i\}$.

This reconfiguration algorithm, applied to the trees $T_I$ and $T_F$ from Figure~\ref{fig:phase2-setup}, is depicted in Figure~\ref{fig:phase2}.


\remove{
To simplify the description of the algorithm, imagine $T$ colored red and $T'$ colored blue. 
\therese{In my opinion we should do introduce this color-scheme much earlier, in the intro, and use it throughout.}
Refer to Figure~\ref{fig:phase2-setup}.
Recall that
$v_1, \ldots, v_n$ is the ordering of vertices of $T$ from lowest to highest.
Write $r_i$ for the (unique) red edge in $T'$ going upward from $v_i$ ($i=1,\dots,n-1$) and
write $b_i$ for the (unique) blue edge in $T$ going down from $v_i$ ($i=2,\dots,n$).

For $i=1, \ldots, n$ we will create $T_i$ as follows.  Set $T_1=T'$.
For $i\geq 2$, if $b_i$ is in $T_{i-1}$, then we simply set $T_{i} := T_{i-1}$.  
Otherwise, consider the cycle $\gamma_i$ that exists in $T_{i-1}\cup \{b_i\}$.   This must
contain some unhappy red edges, else there would be a cycle within the blue edges.   
Among all those unhappy red edges, let $r^*$ be the one with the lowest bottom endpoint and
set $T_i := T_{i-1} \cup \{b_i\} - \{r^*\}$.
This reconfiguration algorithm, applied to the trees $T$ and $T'$ from Figure~\ref{fig:phase2-setup}, is depicted in Figure~\ref{fig:phase2}. 
}

\begin{figure}[htp]
    \centering
    \includegraphics[page=2,scale=0.7]{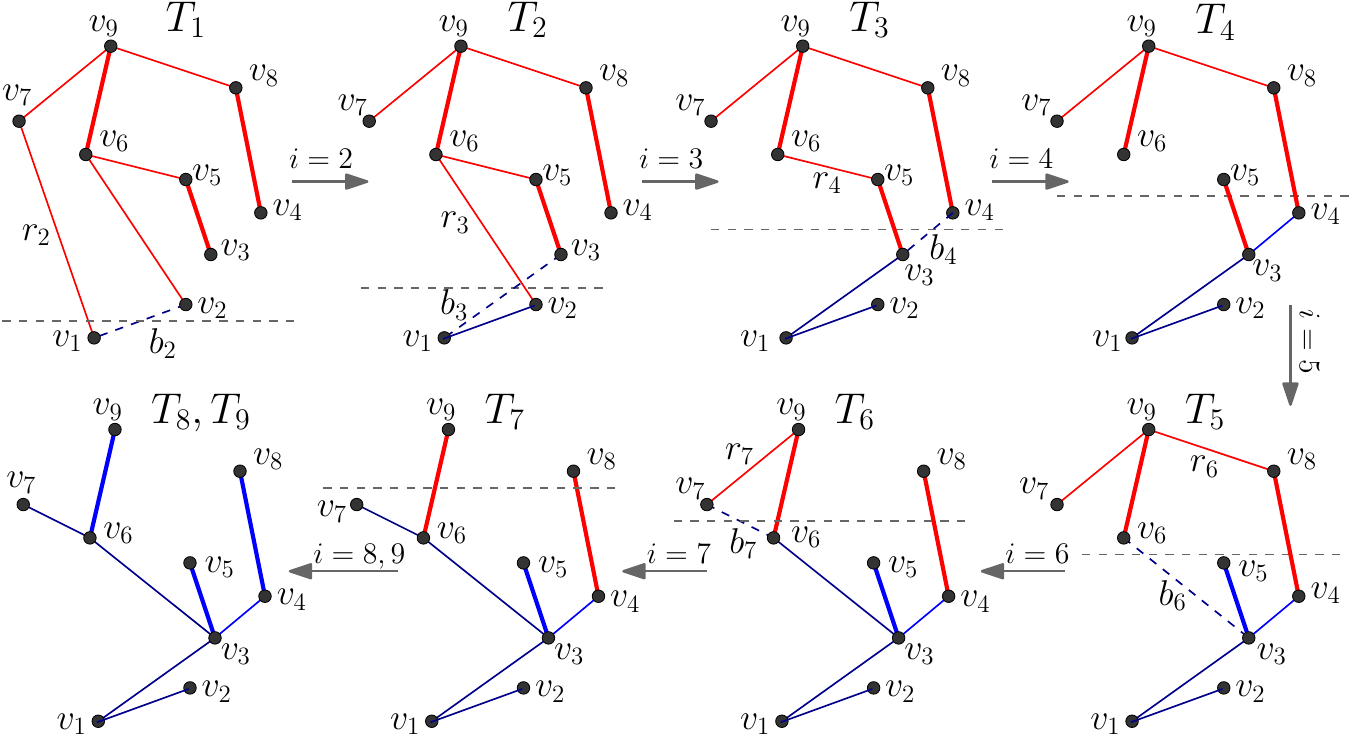}
    \caption{Phase 2: Reconfiguring downward tree $T_1$ into upward tree $T_{9}$.
    The dashed horizontal line separates $B_i$ (below, and drawn with blue edges) from $R_i$ (above).  Happy edges are drawn thick. 
    $T_1$ has unhappy connector-edge $r_2$.
	$T_2$ has the unhappy connector-edge $r_3$ crossing $b_3$. $T_4$, $T_6$ and $T_7$ have two happy connector-edges each. 
    }
    \label{fig:phase2}
\end{figure}

\remove{
Observe that $T_i$ (for $i\geq 1$) satisfies the following:
\begin{itemize}
\item[(A)] $T_i$ is a spanning tree that contains all happy edges. 
\item[(B)] All edges of $T_i$ incident to $v_{i+1},\dots,v_n$ are red.   

	[For future reference we split these edges into two groups: \defn{connector-edges} 
	are edges of the form $(v_h,v_j)$ with $h\leq i < j$, while $R_i$ denotes
	the set of edges within $v_{i+1},\dots,v_n$.]
\item[(C)] Edges $b_2,\dots,b_i$ all belong to $T_i$.  

	[For future reference, we use $B_i$ to denote the subtree
	formed by these edges.   Note that $B_i$ is connected since every vertex $v_h$
	for $1<h\leq i$ is connected to a lower-indexed vertex via $b_h$.]
\end{itemize}
}

\begin{theorem}\label{theo:twocanonical}
	Given a red downward tree $T_I$ and a blue upward tree $T_F$ on a general point set, the reconfiguration algorithm described above flips $T_I$ to $T_F$ using $|T_I \setminus T_F|$ perfect flips. 
\label{thm:trees-up-down}
\end{theorem}
\begin{proof}
It is clear that each $T_i$ is a spanning tree, and that each flip is perfect, so the number of flips is $|T_I \setminus T_F|$.  In particular, a happy edge is never removed, so $T_i$ contains all happy edges. 
We must show that each $T_i$ is non-crossing.  By induction, it suffices to show that if step $i \ge 2$ adds edge $b_i$ and removes edge $r_i$, then $b_i$ does not cross any edge of $T_{i-1}$ except possibly $r_i$. We examine the structure of $T_{i-1}$.

Let $B_{i-1}$ be the subtree with edges $b_2, \ldots, b_{i-1}$. 
Note that $B_{i-1}$ is connected.  
By construction, $T_{i-1}$ contains $B_{i-1}$ and all the other edges of $T_{i-1}$ are red. Let $R_{i-1}$ consist of vertices $v_i, \ldots, v_n$ and the edges of $T_{i-1}$ induced on those vertices.  The edges of $R_{i-1}$ are red, and $R_{i-1}$ consists of some connected components (possibly isolated vertices).  
In $T_{i-1}$ each component of $R_{i-1}$ has exactly one red \defn{connector}-edge joining it to $B_{i-1}$.
Thus $T_{i-1}$ consists of $B_{i-1}$, $R_{i-1}$, and the connector-edges for the components of $R_{i-1}$. 

Now consider the flip performed to create $T_i$ by adding edge $b_i$ and removing edge $r_i$. Since $b_i$ is a blue edge, it cannot cross any edge of $B_{i-1}$. 
Since $b_i$'s topmost vertex is $v_i$, $b_i$ cannot cross any edge of $R_{i-1}$.  
Furthermore, $b_i$ cannot cross a happy edge. 
Thus the only remaining possibility is for $b_i$ to cross an unhappy connector-edge. 

We will prove:

\begin{claim}
\label{cl:connector_edge_happy}
If $R_{i-1}$ is disconnected, then all connector-edges are happy.
\end{claim}


Assuming the claim, we only need to show that $b_i$ is non-crossing when $R_{i-1}$ is connected.  Then there is only one connector-edge $r$ joining $R_{i-1} $ to $B_{i-1}$. If $r$ is happy, then $b_i$ cannot cross it.  So assume that $r$ is unhappy. Refer to Figure~\ref{fig:phase2proof}(a).
Now, the cycle $\gamma$ in $T_{i-1} \cup \{b_i\}$ contains $r$, since $b_i$ and $r$ are the only edges between $R_{i-1}$ and $B_{i-1}$.  Of the red edges in $\gamma$,  $r$ is the one with the lowest bottom endpoint.  This implies that $r$ is chosen as $r_{i}$, and removed. 
Therefore $b_i$ does not cross any edges of $T_i$, so $T_i$ is non-crossing.

\begin{figure}[ht]
    \centering
\hspace*{\fill}
\subfigure[~]{\includegraphics[page=6,scale=0.7]{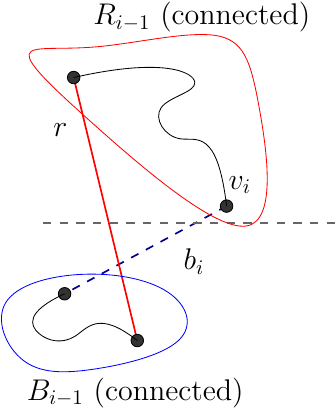}}
\hspace*{\fill}
\subfigure[~]{\includegraphics[page=3,scale=0.7]{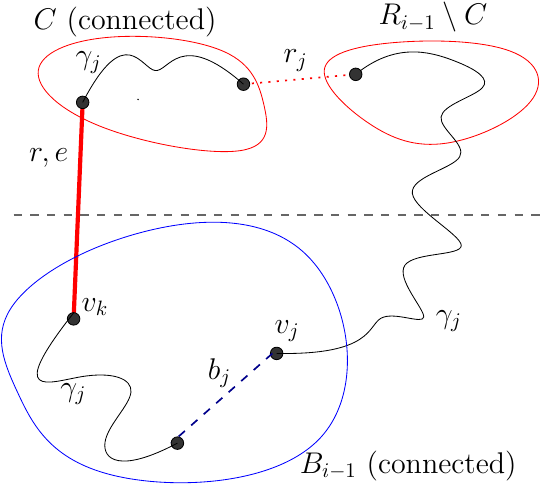}}
\hspace*{\fill}
    \caption{For the proof of
    Theorem~\ref{theorem:upper-bound}: (a) when $R_{i-1}$ is connected; (b) for the proof of Claim~\ref{cl:connector_edge_happy}.}
    \label{fig:phase2proof}
\end{figure}

It remains to prove the claim. Let $C$ be a connected component of $R_{i-1}$, and let $r$ be its connector-edge. Refer to Figure~\ref{fig:phase2proof}(b). We must prove that $r$ is happy. 
In the initial red tree $T_I$, the vertices $v_i, \ldots, v_{n}$ induce a connected subtree, so $C$ and $R_{i-1} \setminus C$ were once connected. Suppose they first became disconnected by the removal of red edge $r_j$ in step $j$ of the algorithm, for some $j<i$.  
Consider the blue edge $b_j$ that was added in step $j$ of the algorithm, and the cycle $\gamma_j$ in $T_{j-1} \cup \{b_j\}$.  Now $\gamma_j$ must contain another edge, call it $e$, with one endpoint in $C$ and one endpoint, $v_k$, not in $C$. 
Note that $e$ is red since it has an endpoint in $C$, and note that $v_k$ is not in $R_{i-1} \setminus C$ otherwise $C$ and $R_{i-1} \setminus C$ would not be disconnected after the removal of $r_j$.  
Therefore $v_k$ must lie in $B_{i-1}$, i.e., $k \le i-1$.  If $e$ is unhappy, then in step $j$ the algorithm would prefer to remove $e$ instead of $r_j$ since $e$ is a red edge in $\gamma_j$ with a lower bottom endpoint.  
So $e$ is happy, which means that the algorithm never removes it, and it is contained in $T_{i-1}$. Therefore $e$ must be equal to $r$, the unique connector-edge in $T_{i-1}$ between $C$ and $B_{i-1}$.  Therefore $r$ is happy. 
\end{proof}